\documentclass[a4paper,UKenglish,cleveref,autoref,numberwithinsect,thm-restate]{lipics-v2021}
\hypersetup{hypertexnames=false}

\usepackage{newfile}
\usepackage{tikz}
\usepackage{bm}
\usepackage[]{todonotes}
\usepackage{totcount}
\usepackage[disableredefinitions]{complexity}
\usepackage{amsmath}
\usepackage{mathrsfs}
\usepackage{amssymb}
\usepackage{mathtools} %
\usetikzlibrary{calc,decorations,patterns,arrows,arrows.meta,decorations.pathmorphing}

\usepackage[ruled,noend,linesnumbered]{algorithm2e}

\nolinenumbers
\hideLIPIcs

\makeatletter
\newcommand{\oset}[3][0ex]{%
  \mathrel{\mathop{#3}\limits^{
    \vbox to#1{\kern-2\ex@
    \hbox{$\scriptstyle#2$}\vss}}}}
\makeatother

\def\nat{\mathbb{N}}

\def\set#1{{\{ #1 \}}}

\newcommand{\ltr}[1]{\mathsf{#1}}%

\def\pre{\leq_{\mathsf{pre}}}

\def\dip{{\mathsf{dip}}}
\def\offset{{\mathsf{offset}}}

\newcommand{\inside}[1]{\mathsf{inside}(#1)}
\newcommand{\outside}[1]{\mathsf{outside}(#1)}
\newcommand{\deriv}{\Rightarrow}

\newcommand{\rev}{\mathsf{rev}}
\newcommand{\markl}{\mathsf{L}}
\newcommand{\markr}{\mathsf{R}}

\newcommand{\asyncp}{\mathscr{P}}
\newcommand{\cT}{\mathcal{T}}%
\def\cG{\mathcal{G}}%
\newcommand{\cV}{\mathcal{V}}%
\newcommand{\multiset}[1]{{\mathbb{M}[ #1 ]}} %
\def\card#1{\lvert {#1} \rvert} %
\def\multi#1{{[\![ #1 ]\!]}} %
\def\mmap{\mathbf{m}} %

\newcommand{\vecz}{\mathbf{0}}%
\newcommand{\bu}{\mathbf{u}}

\newcommand{\bx}{\mathbf{x}}
\newcommand{\by}{\mathbf{y}}

\newcommand{\be}{\mathbf{e}}%

\newcommand{\VASS}{\mathsf{VASS}}

\newcommand{\NFA}{\mathsf{NFA}}
\newcommand{\dsNFA}{\mathsf{dsNFA}}
\newcommand{\dsVASS}{\mathsf{dsVASS}}
\newcommand{\CFG}{\mathsf{CFG}}
\newcommand{\ECFG}{\mathsf{ECFG}}%
\newcommand{\sCFG}{\mathsf{sCFG}}%
\newcommand{\sECFG}{\mathsf{sECFG}}%

\newcommand{\cE}{\mathcal{E}} %
\newcommand{\cB}{\mathcal{B}} %
\newcommand{\cA}{\mathcal{A}} %

\newcommand{\hatie}{\hat{I}_{\varepsilon}}
\newcommand{\barI}{\overline{I}}
\newcommand{\bari}{\overline{i}}

\newcommand{\Parikh}{\mathsf{Parikh}}

\DeclareDocumentCommand{\deriv}{O{}}{\mathrel{\Rightarrow_{#1}}}
\DeclareDocumentCommand{\derivs}{O{}}{\mathrel{\oset{*}{\Rightarrow}_{#1}}}

\newcommand{\Z}{\mathbb{Z}}
\newcommand{\Bounded}{\mathscr{B}}
\newcommand{\Decomps}{\mathscr{M}}
\newcommand{\AdmDecomps}{\mathscr{A}}
\newcommand{\calV}{\mathcal{V}}
\newcommand{\calG}{\mathcal{G}}
\newcommand{\calA}{\mathcal{A}}
\newcommand{\calB}{\mathcal{B}}

\newcommand{\N}{\mathbb{N}}
\newcommand{\calT}{\mathcal{T}}
\newcommand{\cH}{\mathcal{H}}%

\newcommand{\ThetaL}{\Theta_{\markl}}
\newcommand{\ThetaR}{\Theta_{\markr}}
\newcommand{\ThetaLR}{\Theta_{\markl,\markr}}
\newcommand{\offsetL}{\offset\markl}
\newcommand{\offsetR}{\offset\markr}

\newcommand{\dclr}[1]{#1\mathop{\downarrow_{\markl,\markr}}}

\newcommand{\subword}{\preccurlyeq}
\newcommand{\Dyck}{\mathsf{Dyck}}
\newcommand{\sdyck}{\trianglelefteq}
\newcommand{\extsw}{\sqsubseteq}

\newcommand{\dc}[1]{#1\mathord{\downarrow}}

\newcommand{\mindip}{\mathsf{mindip}}

\newcommand{\ioff}{\mathsf{o}}
\newcommand{\idip}{\mathsf{d}}
\newcommand{\imis}{\mathsf{m}}

\makeatletter
\pgfset{
  /pgf/decoration/randomness/.initial=2,
  /pgf/decoration/wavelength/.initial=100
}
\pgfdeclaredecoration{sketch}{init}{
  \state{init}[width=0pt,next state=draw,persistent precomputation={
    \pgfmathsetmacro\pgf@lib@dec@sketch@t0
  }]{}
  \state{draw}[width=\pgfdecorationsegmentlength,
  auto corner on length=\pgfdecorationsegmentlength,
  persistent precomputation={
    \pgfmathsetmacro\pgf@lib@dec@sketch@t{mod(\pgf@lib@dec@sketch@t+pow(\pgfkeysvalueof{/pgf/decoration/randomness},rand),\pgfkeysvalueof{/pgf/decoration/wavelength})}
  }]{
    \pgfmathparse{sin(2*\pgf@lib@dec@sketch@t*pi/\pgfkeysvalueof{/pgf/decoration/wavelength} r)}
    \pgfpathlineto{\pgfqpoint{\pgfdecorationsegmentlength}{\pgfmathresult\pgfdecorationsegmentamplitude}}
  }
  \state{final}{}
}
\tikzset{xkcd/.style={decorate,decoration={sketch,segment length=0.5pt,amplitude=0.5pt}}}
\makeatother

\title{Checking Refinement of Asynchronous Programs against Context-Free Specifications} 
\titlerunning{Checking Refinement of Asynchronous Programs} 

\author{Pascal Baumann}{Max Planck Institute for Software Systems (MPI-SWS), Germany}{pbaumann@mpi-sws.org}{https://orcid.org/0000-0002-9371-0807}{}

\author{Moses Ganardi}{Max Planck Institute for Software Systems (MPI-SWS), Germany}{ganardi@mpi-sws.org}{https://orcid.org/0000-0002-0775-7781}{}

\author{Rupak Majumdar}{Max Planck Institute for Software Systems (MPI-SWS), Germany}{rupak@mpi-sws.org}{https://orcid.org/0000-0003-2136-0542}{}

\author{Ramanathan S. Thinniyam}{Max Planck Institute for Software Systems (MPI-SWS), Germany}{thinniyam@mpi-sws.org}{https://orcid.org/0000-0002-9926-0931}{}

\author{Georg Zetzsche}{Max Planck Institute for Software Systems (MPI-SWS), Germany}{georg@mpi-sws.org}{https://orcid.org/0000-0002-6421-4388}{}

\authorrunning{P. Baumann, M. Ganardi, R. Majumdar, R. S. Thinniyam, and G. Zetzsche}

\Copyright{Pascal Baumann, Moses Ganardi, Rupak Majumdar, Ramanathan S. Thinniyam, and Georg Zetzsche} 



\category{Track B: Automata, Logic, Semantics, and Theory of Programming}

\relatedversion{} %
\EventEditors{Kousha Etessami, Uriel Feige, and Gabriele Puppis}
\EventNoEds{3}
\EventLongTitle{50th International Colloquium on Automata, Languages, and Programming (ICALP 2023)}
\EventShortTitle{ICALP 2023}
\EventAcronym{ICALP}
\EventYear{2023}
\EventDate{July 10--14, 2023}
\EventLocation{Paderborn, Germany}
\EventLogo{}
\SeriesVolume{261}
\ArticleNo{109}

\funding{\flag[3cm]{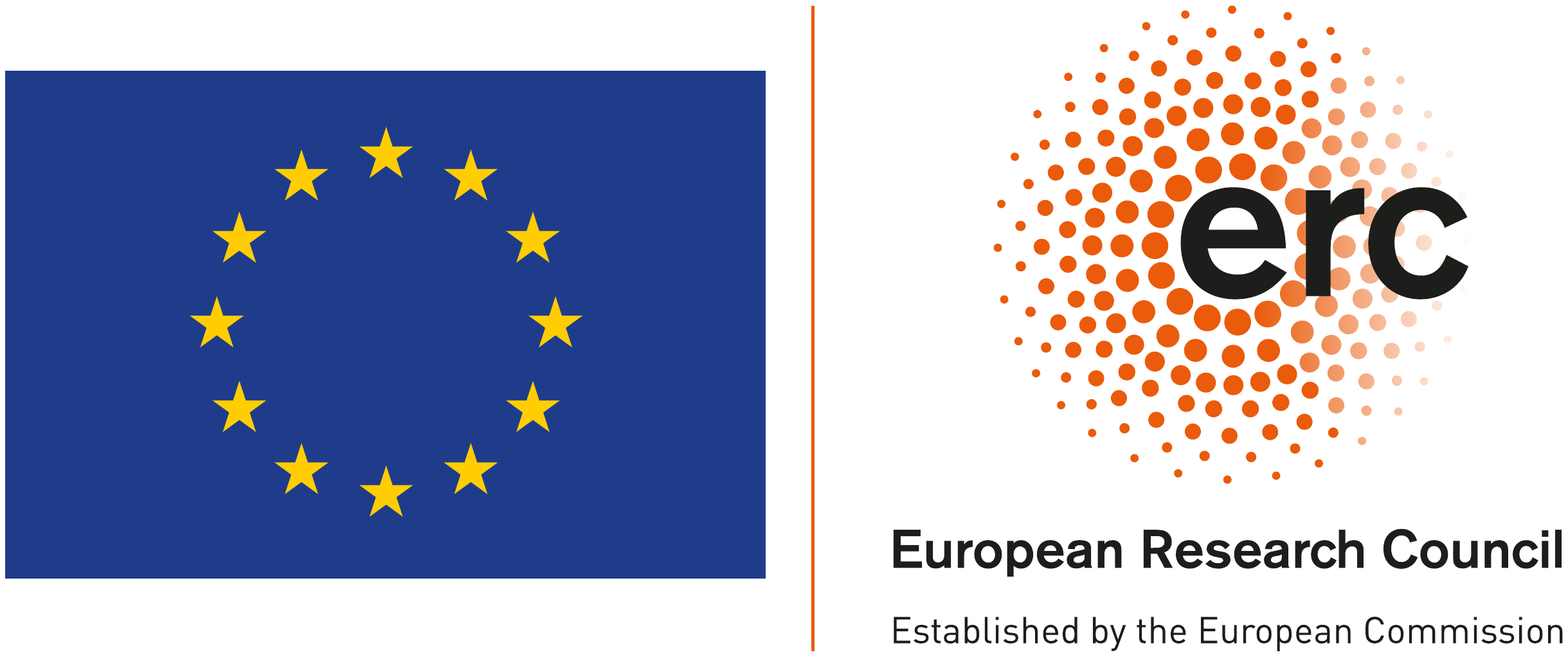}Funded by the European Union (ERC, FINABIS, 101077902). Views and opinions expressed are however those of the author(s) only and do not necessarily reflect those of the European Union or the European Research Council Executive Agency. Neither the European Union nor the granting authority can be held responsible for them.
Partially funded by the DFG project 389792660 TRR 248–CPEC.
}

\begin{document}  
\maketitle

\begin{abstract}
In the language-theoretic approach to refinement verification, we check that the 
language of traces of an implementation all belong to the language of a specification.
We consider the refinement verification problem for asynchronous programs against specifications given by a Dyck
language.
We show that this problem is $\EXPSPACE$-complete---the same complexity as that of language emptiness and for
refinement verification against a regular specification.
Our algorithm uses several technical ingredients.
First, we show that checking if the coverability language of a succinctly described
vector addition system with states (VASS) is contained in a Dyck language
is $\EXPSPACE$-complete.
Second, in the more technical part of the proof, we define an ordering on words and show a downward closure construction that 
allows replacing the (context-free) language of each task in an asynchronous program by a regular language.
Unlike downward closure operations usually considered in infinite-state verification, our ordering is not a well-quasi-ordering, and we have
to construct the regular language ab initio.
Once the tasks can be replaced, we show a reduction to an appropriate VASS and use our first ingredient. 
In addition to the inherent theoretical interest, refinement verification with Dyck specifications captures common
practical resource usage patterns based on reference counting, for which few algorithmic techniques were known.
\end{abstract}

\section{Introduction} %
\label{sec:introduction}

Asynchronous programs are a common programming idiom for multithreaded shared memory concurrency.
An asynchronous program executes tasks atomically; each task is a sequential recursive program that can
read or write some shared state, emit events (such as calling an API), and, in addition, can spawn an arbitrary number of new tasks
for future execution.
A cooperative scheduler iteratively picks a previously spawned task and executes it atomically to completion.
Asynchronous programs occur in many software systems with stringent correctness requirements.
At the same time, they form a robustly decidable class of infinite-state systems closely aligned with other concurrency models.
Thus, algorithmic verification of asynchronous programs has received a lot of attention from both theoretical 
and applied perspectives \cite{SenViswanathan06,JhalaM07,GantyM12,DesaiGM14,DesaiQ17,KraglEHMQ20,GavranNKMV15,KraglQH18,MajumdarTZ21}.

We work in the language-theoretic setting, where we treat asynchronous programs as generators of languages, and reduce verification
questions to decision problems on these languages.
Thus, an execution of a task yields a word over the alphabet of its events and task names. 
An execution of the asynchronous program concatenates the words of executing tasks and further ensures that any task executing
in the concatenation was spawned before and not already executed.
The trace of an execution projects the word to the alphabet of events and the language of the program is the set of all traces. 
With this view, reachability or safety verification questions reduce to language emptiness, and
refinement verification reduces to language inclusion of a program in a given specification language over the alphabet of events.

We consider the language inclusion problem for asynchronous programs when the specification language is given by a Dyck language.
Our main result shows that this problem is $\EXPSPACE$-complete.
The language emptiness problem for asynchronous programs, as well as language inclusion in a regular language, are already $\EXPSPACE$-complete~\cite{GantyM12}.
Thus, there is no increase in complexity even when the specifications are Dyck languages.
However, as we shall see below, our proof of membership in $\EXPSPACE$ requires several new ingredients.

In addition to the inherent language-theoretic interest, the problem is motivated by the practical ``design pattern'' of reference counting
and barrier synchronization in concurrent event-driven programs.
In this pattern, each global shared resource maintains a counter of \emph{how many} processes have access to it.
Before working with the shared resource, a task acquires access to the resource by incrementing a counter (the reference count).
Later, a possibly different task can release the resource by decrementing the reference count.
When the count is zero, the system can garbage collect the resource. 
For example, device drivers in the kernel maintain such reference counts, and there are known bugs arising out of
incorrect handling of reference counts \cite{QadeerWu}. 
Here is a small snippet that shows the pattern in asynchronous code:
\vspace{-0.3cm}
\[
\begin{array}{ll}
\texttt{start}: & \{\ t\ :=\ \texttt{inc}();\ \texttt{if}\ (t)\ \texttt{spawn}(\texttt{work}); \}\\
                & //\ \mbox{\small arbitrarily many requests may start concurrently}\\
\texttt{work}: & \{ \mbox{ \small in this code, we can assert that the reference count is positive }; \\
               & \ \ \ \texttt{spawn}(\texttt{cleanup}); \}\\
\texttt{cleanup}: & \{\  \texttt{dec}();\ \texttt{if}\ \texttt{zeroref}()\ \{ \mbox{ \small garbage collect the resource }\}\}\\
\end{array}
\]
Here, $\texttt{inc}$ and $\texttt{dec}$ increment and decrement the reference count associated with a shared resource,
$\texttt{inc}$ succeeds if the resource has not been garbage collected.
$\texttt{spawn}$ starts a new task, and $\texttt{zeroref}$ checks if the reference count is zero.
There are three tasks, $\texttt{start}$, $\texttt{work}$, and $\texttt{cleanup}$; each invocation of a task executes atomically.
Initially, an arbitrary number of $\texttt{start}$ tasks are spawned.

Our goal is to ensure the device is not garbage collected while some instance of $\texttt{work}$ is pending.
Intuitively, the reason for this is clear: each $\texttt{work}$ is spawned by a previous $\texttt{start}$ that takes a reference count and this reference
is held until a later $\texttt{cleanup}$ runs.
However, it is difficult for automated model checking tools to perform this reasoning, and existing techniques require manual annotations
of invariants \cite{GavranNKMV15,KraglEHMQ20}.
Dyck languages allow specifying correct handling of reference counts \cite{BaumannGMTZ23}, 
and our algorithm provides as a special case an algorithmic analysis of correct reference counting for asynchronous programs.

Since there is a simple reduction from language emptiness to inclusion, we immediately inherit $\EXPSPACE$-hardness.
Let us therefore focus on the challenges in obtaining an $\EXPSPACE$ upper bound.
The $\EXPSPACE$ algorithm for language emptiness proceeds as follows (see \cite{GantyM12,MajumdarTZ21}).
First, we can ignore the alphabet of events and only consider words over the alphabet of task names.
Second, we notice that (non-)emptiness is preserved if we ``lose'' some spawns along an execution;
this allows us to replace the language of each task by its downward closure.
By general results about well-quasi orderings, the downward closure is a regular language which, moreover, has a succinct representation.
Thus, we can reduce the language emptiness problem to checking (coverability) language emptiness of an associated vector addition system with states ($\VASS$).
This problem can be solved in $\EXPSPACE$, by a result of Rackoff \cite{Rackoff78}.

Unfortunately, this outline is not sufficient in our setting. 
First, unlike for language emptiness or regular language inclusion, we cannot simply replace tasks with their downward closures (w.r.t.\ the subword ordering).
While we can drop spawns as before, dropping letters from the event alphabet does not preserve membership in a Dyck language.
Second, even if each handler is regular, we are left with checking if a $\VASS$ language is contained in a Dyck language.
We provide new constructions to handle these challenges.

Our starting point is the characterization of inclusion in Dyck languages \cite{RitchieSpringsteel1972}: A language~$L$ is not included in a Dyck language
if and only if there is a word $w \in L$ with either an \emph{offset violation} (number of open brackets does not match the number of closed brackets), a \emph{dip violation} (some prefix with more closed brackets than open ones),
or a \emph{mismatch violation} (an open bracket of one kind matched with a closed bracket of a different kind).

\subparagraph{Checking $\VASS$ Language Inclusion}
Our first technical construction shows how to check language inclusion of a $\VASS$ coverability language in a Dyck language in $\EXPSPACE$. (In a \emph{coverability language}, acceptance is defined by reaching a final control state.)
In fact, our result carries over when the control states of the $\VASS$ are \emph{succinctly represented}, for example by using transducers and binary encodings of numbers. 

We first check that the $\VASS$ language is \emph{offset-uniform}, that is,
every word in the language has exactly the same offset (difference between open brackets and closed brackets),
and that this offset is actually zero.
(If this condition is not true, there is already an offset violation.)
We show that the offset of every prefix of a word in any offset-uniform $\VASS$ language
is bounded by a doubly exponential number,
and therefore, this number can be tracked by adding
double exponentially bounded counters (as in Lipton's construction \cite{lipton1976reachability}) in the $\VASS$ itself.
Moreover, we can reduce the checking of dip or mismatch violations to finding a \emph{marked Dyck factor}: an infix of the form
$\# w \bar{\#}$ for a Dyck word $w$.
Finally, for offset-uniform $\VASS$, finding a marked Dyck factor reduces to coverability in succinctly represented $\VASS$, which can be checked in $\EXPSPACE$~\cite{BaumannMTZ22}.
Offset uniformity is important---finding a marked Dyck factor in an arbitrary $\VASS$ language is equivalent to $\VASS$ reachability, which is Ackermann-complete \cite{CzerwinskiPNAckermann2021,lerouxReachabilityProblemPetri2022}.
In fact, checking whether a given $\VASS$ language is included in the set of \emph{prefixes} of the one-letter Dyck language is already equivalent to $\VASS$ reachability (see the long version of the paper for a proof).

A consequence of our result is that given a $\VASS$ coverability language $K$
and a \emph{reachability language} (i.e.\ acceptance requires all counters to be zero in the end)
$L$ of a \emph{deterministic} $\VASS$, deciding whether $K\subseteq L$ is
$\EXPSPACE$-complete. This is in contrast (but not in contradiction\footnote{For general $\VASS$, every coverability language is also a
	reachability language. However, \emph{deterministic} $\VASS$ with reachability
acceptance cannot accept all coverability languages.}) to recent
Ackermann-completeness results for settings where both $K$ and $L$ are drawn from
subclasses of $\VASS$ coverability languages~\cite{CzerwinskiH22}.

\subparagraph{Downward Closure of Tasks}
Next, we move to asynchronous programs.
We define a composite ordering on words that is a combination of two different orderings: the subword ordering for task names, and the syntactic preorder on the events projected to a single set $\{ x,\bar{x}\}$ of Dyck letters. In our case, the latter means a word $u$ is less than $v$ iff they both have the same offset, but $v$ has at most the dip of $u$.  The composite order is defined so as to preserve the existence of marked Dyck factors.
In contrast to the subword ordering, this (composite) ordering is not a well-quasi-ordering (since, e.g., $\bar{x}x,\bar{x}\bar{x}xx,\bar{x}\bar{x}\bar{x}xxx, \ldots$ forms an infinite descending chain).
Nevertheless, our most difficult technical construction shows that for any context-free language (satisfying an assumption, which we call tame-pumping)
there exists a regular language with the same downward closure in this ordering. The case of general context-free languages reduces to this special case since the presence of a non-tame pump immediately results in a Dyck-violation and can easily be detected in $\PSPACE$.
For the tame-pumping grammars, a succinct description of the corresponding automaton can be computed in $\PSPACE$.
This key observation allows us to replace the context-free languages of tasks with regular sets, and thereby 
reduce the problem to checking $\VASS$ language inclusion.

\subparagraph{Related Work} Language inclusion in Dyck languages is a well-studied problem.
For example, inclusion in a Dyck language can be checked in polynomial time for context-free languages \cite{TozawaM07} or for ranges of two-copy tree-to-string transducers \cite{Linear-Tree-Transducers-lobelthesis}.
Our work extends the recent result that the language noninclusion problem for context-bounded multi-pushdown systems in Dyck languages is $\NP$-complete \cite{BaumannGMTZ23}.
Our result is complementary to that of \cite{BaumannGMTZ23}: their model considers a \emph{fixed} number of threads but allows the threads to be interrupted and context-switched a fixed number of times.
In contrast, we allow dynamic spawning of threads but assume each thread is atomically run to completion.
A natural open question is whether our results continue to hold if threads can be interrupted up to a fixed number of times.

Inclusion problems have recently also been studied when both input languages
are given as $\VASS$ coverability languages~\cite{CzerwinskiH22}. Since in our
setting, the supposedly larger language is always a Dyck language (which is not
a coverablity $\VASS$ language), those results are orthogonal.

\section{Language-Theoretic Preliminaries} %
\label{sec:preliminaries}

\subparagraph*{General Definitions}

We assume familiarity with basic language theory, see the textbook~\cite{Kozen} for more details.
For an alphabet $\Sigma\subseteq\Theta$, let $\pi_\Sigma\colon\Theta^*\to\Sigma^*$
denote the projection onto $\Sigma^*$.
In other words, for $w\in\Theta^*$, the word $\pi_\Sigma(w)$ is obtained from $w$
by deleting every occurrence of a letter in $\Theta\setminus\Sigma$.
If $\Sigma$ contains few elements, e.g.\ $\Sigma = \{x,y\}$, then
instead of writing $\pi_{\{x,y\}}$ we also write $\pi_{x,y}$, leaving out the set brackets.
We write $|w|_\Sigma$ for the number of occurrences of letters $x \in \Sigma$ in $w$,
and similarly $|w|_x$ if $\Sigma = \{x\}$.
%
%
%
%
%
%
%
%
%
%

\subparagraph*{Context-Free Languages}
A \emph{context-free grammar} ($\CFG$) $\cG = (N,\Theta,P,S)$ consists of
an alphabet of \emph{nonterminals} $N$,
an alphabet of \emph{terminals} $\Theta$ with $N \cap \Theta = \emptyset$,
a finite set of \emph{productions} $P \subseteq N \times (N \cup \Theta)^*$, and
the start symbol $S \in N$.
We usually write $A \rightarrow v$ to denote a production $(A,v) \in P$.
The size of the $\CFG$ $\cG$ is defined as
$|\cG| = \sum_{A \rightarrow v \in P} (|v| + 1)$.
We denote the \emph{derivation relation} by $\deriv_\cG$ and its reflexive, transitive closure by $\derivs_\cG$.
We drop the subscript $\cG$ if it is clear from the context.
We also use \emph{derivation trees} labelled by $N \cup \Theta$ for derivations
of the form $A \derivs w$ for some $A \in N$.
Here we start with the root labelled by $A$, and whenever we apply a production
$B \rightarrow v$ with $v = a_1 \ldots a_n$, we add $n$ children labelled by $a_1,\ldots,a_n$
(in that order from left to right) to a leaf labelled by $B$. A \emph{pump} is a derivation of the form $A \derivs u A v$ for some nonterminal $A$. A derivation tree which is pumpfree,
i.e., in which no path contains multiple occurrences of the same nonterminal, is referred to as a \emph{skeleton}. 
We will often see an arbitrary derivation tree as one which is obtained by inserting pumps into a skeleton.

The \emph{language} $L(\cG, A)$ of $\cG$ starting from nonterminal $A \in N$ contains
all words $w \in \Theta^*$ such that there exists a derivation $A \derivs_\cG w$.
The language of $\cG$ is $L(\cG) = L(\cG,S)$.
A \emph{context-free language} ($\CFL$) $L$ is a language for which there exists
a $\CFG$ $\cG$ with $L = L(\cG)$.

A $\CFG$ $\cG = (N,\Theta,P,S)$ is said to be in \emph{Chomsky normal form}
if all of its productions have one of the forms
$A \rightarrow BC$, $A \rightarrow a$, or $S \rightarrow \varepsilon$,
where $B,C \in N\setminus\set{S}$, $a \in \Theta$, and the last form only occurs
if $\varepsilon \in L(\cG)$.
It is well known that every $\CFG$ can be transformed in polynomial time into one
in Chomsky normal form with the same language.

An \emph{extended} context-free grammar ($\ECFG$) $\cG = (N,\Theta,P,S)$ is a $\CFG$,
which may additionally have productions of the form
$A \rightarrow \Gamma^* \in P$ for some alphabet $\Gamma \subseteq \Theta$.
Productions of this form induce derivations $uAs \deriv_\cG uvs$,
where $u,s \in (N \cup \Theta)^*$ and $v \in \Gamma^*$.
Chomsky normal form for $\ECFG$ is defined as for $\CFG$, but also allows productions
of the form $A \rightarrow \Gamma^*$.
An $\ECFG$ can still be transformed into Chomsky normal form using the same algorithm
as for a $\CFG$, treating expressions $\Gamma^*$ like single terminal symbols.
Since the extended productions can be simulated by conventional $\CFG$ productions,
the language of an $\ECFG$ is still a $\CFL$.

\subparagraph*{Dyck Language}
Let $X$ be an alphabet and let $\bar{X} = \{\bar{x} \mid x \in X\}$ be a disjoint copy of $X$.
The \emph{Dyck language (over $X$)} $\Dyck_X \subseteq (X \cup \bar X)^*$ is defined by the following
context-free grammar:
\[
	S \to \varepsilon \mid S \to SS \mid S \to x S \bar x \quad \text{for } x \in X.
\]
Let $\Theta \supseteq X \cup \bar X$ be an alphabet.
For $w \in \Theta^*$ we define $\offset(w) = |w|_X-|w|_{\bar X}$. 
A language $L\subseteq \Theta^*$ is called \emph{offset-uniform} if for
any $u,v\in L$, we have $\offset(u)=\offset(v)$.

The \emph{dip} of $w \in \Theta^*$ is defined as
$\dip(w)=\max \{-\offset(u) \mid u \text{ is a prefix of } w \}$.
We define $e(w) = (\dip(w), \offset(w))$.
Observe that for $w \in (X \cup \bar{X})^*$ with $|X| = 1$
we have $w \in \Dyck_X$ if and only if $e(w) = (0,0)$.

A language $L\subseteq (X \cup \bar X)^*$ is \emph{not} included in $\Dyck_X$
if and only if there exists a word $w \in L$
that satisfies one of the following violation conditions \cite{RitchieSpringsteel1972}: 
\begin{description}
\item[(OV)] an \emph{offset violation} $\offset(w) \neq 0$,
\item[(DV)] a \emph{dip violation}, where $\dip(w) > 0$, i.e., there is a prefix $u$ of $w$
with $\offset(u) < 0$, or
\item[(MV)] a \emph{mismatch violation}, where there exists a pair $x, \bar{y}$
(for some $x \neq y$) of \emph{mismatched} letters in $w$,
i.e., $w$ contains an infix $x v \bar{y}$ where $e(v) = (0,0)$.
\end{description} 
For example, $w_1=x\bar{x}\bar{x}x$ has a dip violation due to the
prefix $u=x\bar{x}\bar{x}$; $w_2=xx\bar{x}$ has an offset violation
and $w_3=x x \bar{x} \bar{y}$ has a mismatch violation.

\section{Asynchronous Programs}
\label{sec:asyncp}

An \emph{asynchronous program} \cite{GantyM12}, henceforth simply called a \emph{program},
is a tuple $\asyncp = (Q,$ $\Sigma,$ $\Gamma,$ $\cG,$ $\Delta,$ $q_0,$ $q_f,$ $\gamma_0)$,
where
$Q$ is a finite set of \emph{global states},
$\Sigma$ is an alphabet of \emph{event letters},
$\Gamma$ is an alphabet of \emph{handler names} with $\Sigma \cap \Gamma = \emptyset$,
$\cG$ is a $\CFG$ over the terminal symbols $\Sigma\cup\Gamma$,
$\Delta$ is a finite set of transition rules (described below),
$q_0 \in Q$ is the \emph{initial state},
$q_f \in Q$ is the \emph{final state}, and
$\gamma_0$ is the \emph{initial handler}.

Transition rules in $\Delta$ are of the form $q \xhookrightarrow{a,A} q'$, where
$q, q' \in Q$ are global states, $a \in \Gamma$ is a handler name,
and $A$ is a nonterminal symbol in $\cG$.

Let $\multiset{S}$ denote the set of all multisets of elements from the set $S$. A \emph{configuration} $(q,\mmap) \in Q \times \multiset{\Gamma}$ of $\asyncp$
consists of a global state $q$ and a multiset $\mmap: \Gamma\rightarrow\nat$ of pending handler instances.
The \emph{initial} configuration of $\asyncp$ is $c_0 = (q_0,\multi{\gamma_0})$,
where $\multi{\gamma_0}$ denotes the singleton multiset containing $\gamma_0$.
A configuration is considered \emph{final} if its global state is $q_f$.
The rules in $\Delta$ induce a transition relation on configurations of $\asyncp$:
We have $(q,\mmap) \xrightarrow{w} (q',\mmap')$ iff there is a rule
$q \xhookrightarrow{a,A} q' \in \Delta$ and a word $u \in L(\cG, A)$ such that
$\pi_\Sigma(u) = w$ and $\mmap' = (\mmap \ominus \multi{a}) \oplus \Parikh(\pi_\Gamma(u))$, where $\mmap''= \mmap \oplus \mmap'$ is the multiset which satisfies $\mmap''(a)=\mmap'(a)+\mmap(a)$ for each $a \in \Gamma$. Similarly $\mmap''= \mmap \ominus \mmap'$ is the multiset which satisfies $\mmap''(a)=\mmap'(a)-\mmap(a)$ for each $a \in \Gamma$ with the implicit assumption that $\mmap'(a) \geq \mmap(a)$.
Here, $\Parikh(w): \Gamma\rightarrow \nat$ is the \emph{Parikh image} of $w$ that
maps each handler in $\Gamma$ to its number of occurrences in $w$.
Note that the transition is feasible only if $\mmap$ contains at least one instance of the handler $a$.

Intuitively, a program consists of a set of atomic event handlers that communicate over a shared global state $Q$.
Each handler is a piece of sequential code that generates a word over a set of events $\Sigma$ and, in addition,
posts new instances of handlers from $\Gamma$.
A configuration $(q,\mmap)$ represents the current value of the shared state $q$ and a task buffer $\mmap$ containing
the posted, but not yet executed, handlers.
At each step, a scheduler non-deterministically picks and removes a handler from the multiset of posted handlers and ``runs'' it.
Running a handler changes the global state and produces a sequence of events over $\Sigma$ as well as a multiset of newly posted handlers.
The newly posted handlers are added to the task buffer.

We consider asynchronous programs as generators of words over the set of events.
A \emph{run} of $\asyncp$ is a finite sequence of configurations
$c_0 = (q_0,\multi{\gamma_0}) \xrightarrow{w_1} c_1
\xrightarrow{w_2} \ldots \xrightarrow{w_\ell} c_\ell$.
It is an \emph{accepting} run if it ends in a final configuration.

The \emph{language} of $\asyncp$ is defined as
\[ L(\asyncp)=\{ w \in \Sigma^* \mid w=w_1\cdots w_\ell, \text{ there is an accepting run }
c_0 \xrightarrow{w_1} %
\ldots \xrightarrow{w_\ell} c_\ell\}. \]
The size of the program $\asyncp$ is defined as
$|\asyncp| = |Q| + |\cG| + |\Delta|$,
i.e., the combined size of states, grammar, and transitions. 

The \emph{Dyck inclusion problem} for programs asks, given a program $\asyncp$ over a set
$(X\cup \bar X)$ of events, whether every word in $L(\asyncp)$ belongs to the Dyck language $\Dyck_X$.
We show the following main result. 

\begin{theorem}[Main Theorem]
	\label{thm:main}
	Given a program $\asyncp$ with $L(\asyncp) \subseteq (X \cup \bar X)^*$, 
	deciding if $L(\asyncp) \subseteq \Dyck_X$ is $\EXPSPACE$-complete. 
\end{theorem}

$\EXPSPACE$-hardness follows easily from the following result on language emptiness (by simply adding a loop with a letter $\bar{x} \in \bar{X}$ at the final state).   
Therefore, the rest of the paper focuses on the $\EXPSPACE$ upper bound.

\begin{proposition}[Theorem 6.2, Ganty and Majumdar \cite{GantyM12}]
	\label{thm:safetyAsync}
	Given a program $\asyncp$, checking if $L(\asyncp) = \emptyset$ is $\EXPSPACE$-complete.
\end{proposition}

A nonterminal $B$ in the grammar $\cG$ of a program $\asyncp$ is called \emph{useful}
if there exists a run $\rho$ of $\asyncp$ reaching $q_f$
in which there exists a derivation tree containing $B$.
More precisely, there are two successive configurations $(q,\mmap) \xrightarrow{w} (q',\mmap')$
in $\rho$ such that there is a rule $q \xhookrightarrow{a,A} q'$ and a word $u \in L(\cG, A)$ with
$\pi_\Sigma(u) = w$, $\mmap' = (\mmap \ominus \multi{a}) \oplus \Parikh(\pi_\Gamma(u))$,
and $B$ occurs in some derivation tree with root $A$ and yield $u$.
There is a simple reduction from checking if a nonterminal is useful to checking
language emptiness
(see the full version)
so we can check if a nonterminal is useful also in $\EXPSPACE$.
Therefore, in the following, we shall assume that all nonterminals are useful.

\section{Checking Dyck Inclusion for $\VASS$ Coverability Languages}
\label{sec:VASScoverInDyck}

As a first technical construction, we show how to check Dyck inclusion for (succinctly defined) $\VASS$ languages.
We shall reduce the problem for programs to this case.

\subsection{Models: $\VASS$ and Succinct Versions}

\subparagraph*{Vector Addition Systems with States}
A \emph{vector addition system with states} ($\VASS$) is a tuple
$\cV = (Q, \Sigma, I, E,q_0,q_f)$ where
$Q$ is a finite set of \emph{states},
$\Sigma$ is a finite alphabet of \emph{input letters},
$I$ is a finite set of \emph{counters}, 
$q_0\in Q$ is the \emph{initial state}, 
$q_f\in Q$ is the \emph{final state}, 
and
$E$ is a finite set of \emph{edges} of the form $q \xrightarrow{x,\delta} q'$, where
$q, q' \in Q$, $x \in \Sigma \cup \set{\varepsilon}$, and
$\delta \in \set{-1,0,1}^I$.\footnote{A more general definition of $\VASS$
  would allow each transition to add an arbitrary vector over the integers.
  We instead restrict ourselves to the set $\set{-1,0,1}$, since this
  suffices for our purposes, and the $\EXPSPACE$-hardness result by
  Lipton~\cite{lipton1976reachability} already holds for VASS of this form.}

A \emph{configuration} of $\cV$ is a pair $(q,\bu) \in Q \times \multiset{I}$.
The elements of $\multiset{I}$ and $\set{-1,0,1}^I$ can also be seen as vectors of length $|I|$ over $\N$ and $\set{-1,0,1}$, respectively, and we sometimes denote them as such.
The edges in $E$ induce a transition relation on configurations: 
there is a transition $(q, \bu)\xrightarrow{x}(q', \bu')$ if
there is an edge $q\xrightarrow{x,\delta} q'$ in $E$ such that
$\bu'(i) = \bu(i) + \delta(i) \geq 0$ for all $i\in I$.
A \emph{run} of the $\VASS$ is a finite sequence of configurations
$c_0 \xrightarrow{x_1} c_1 \xrightarrow{x_2}
\ldots \xrightarrow{x_\ell} c_\ell$ where $c_0=(q_0,\vecz)$.
A run is said to reach a state $q\in Q$ if the last configuration in the run is of the
form $(q,\mmap)$ for some multiset $\mmap$.
An \emph{accepting} run is a run whose final configuration has state $q_f$.
The \emph{(coverability) language} of $\cV$ is defined as 
\[ L(\cV)=\{ w \in \Sigma^* \mid \text{there exists a run }
(q_0,\vecz)=c_0 \xrightarrow{x_1} %
\ldots \xrightarrow{x_\ell} c_\ell=(q_f,\bu) \text{ with } w=x_1\cdots x_\ell \}. \]
The size of the $\VASS$ $\cV$ is defined as $|\cV| = |I| \cdot |E|$.

\subparagraph*{Models with Succinct Control} 

In this paper we need various models with \emph{doubly succinct} control,
i.e., models with doubly exponentially many states.
Informally speaking, a machine with finite control $\cB$, e.g.\ an $\NFA$ or a $\VASS$, is doubly succinct
if its set of control states is $\Lambda^M$ where $M \in \N$ is an exponential number given in binary encoding,
and $\Lambda$ is a finite alphabet.
The initial and final state of $\cB$ are the states $0^M$ and $1^M$ for some letters $0,1 \in \Lambda$.
Finally, the transitions of $\cB$ are given by \emph{finite-state transducers} $\cT$,
i.e., asynchronous multitape automata recognizing relations $R \subseteq (\Lambda^M)^k$.
For example, a \emph{doubly succinct $\NFA$} ($\dsNFA$ in short)
contains binary transducers $\cT_a$ for each $a \in \Sigma \cup \{\varepsilon\}$ where $\Sigma$ is the input alphabet,
and $\cB$ contains a transition $p \xrightarrow{x} q$ if and only if $(p,q)$ is accepted by $\cT_x$.
A \emph{doubly succinct VASS} ($\dsVASS$, for short) contains binary transducers
$\cT_{x,i},\cT_{x,\bar i},\cT_{x,\varepsilon}$ for each $x \in \Sigma \cup \{\varepsilon\}$ and $i \in I$,
where $I$ is the set of counters.
A state pair $(p,q)$ accepted by $\cT_{x,i}$ specifies a transition $p \xrightarrow{x,\be_i} q$ in $\cB$,
where $\be_i$ only increments counter $i$ and leaves other counters the same.
Similarly $\cT_{x,\bar i}$ and $\cT_{x,\varepsilon}$ specify decrementing transitions
and transitions without counter updates.

Later we will also use \emph{(singly) succinct} $\ECFG$s, which are extended context-free grammars
whose set of nonterminals is $\Lambda^M$ where $M$ is a unary encoded number.
The set of productions is given in a suitable fashion by transducers.
Let us remark that the precise definition of (doubly) succinct automata or grammars is not important for our paper,
e.g.\ one could also use circuits instead of transducers to specify the transitions/productions.

\subsection{Checking Dyck Inclusion for $\dsVASS$}

We prove our first technical contribution: an $\EXPSPACE$ procedure to check non-inclusion of a $\VASS$ language in a Dyck language.
This involves checking if one of (OV), (DV), or (MV) occurs.
We begin by showing how these violations can be detected for a (non-succinct) $\VASS$.

To this end, first we show that offset-uniformity of a $\VASS$ language implies
a doubly exponential bound $B$ on the offset values for prefixes of accepted words
(\cref{thm:vass-boundedness}).
Given an alphabet $X$ and a number $k \in \N$, we define the language
\[ 
\Bounded({X,k}) =\{w\in(X\cup\bar{X})^* \mid \text{for every prefix $v$ of $w$: $|\offset(v)|\le k$}\}. 
\]

\begin{theorem}\label{thm:vass-boundedness}
Let $\calV$ be a $\VASS$ with $L(\calV)\subseteq (X\cup\bar{X})^*$. 
If $L(\calV)$ is offset-uniform, then $L(\calV)\subseteq \Bounded({X,2^{2^{p(|\calV|)}}})$ for some polynomial function $p$. 
\end{theorem}

\begin{proof}
Let $\calV = (Q,X \cup \bar X,I,E,q_0,q_f)$ be a VASS where $L(\calV) \neq \emptyset$ is offset-uniform.
The unique offset of $L(\calV)$ is bounded double exponentially in $|\calV|$ since $L(\calV)$ contains
some word that is at most double exponentially long, a fact that follows from Rackoff's bound on covering runs~\cite{Rackoff78}.
Let $C \subseteq Q \times \multiset{I}$ be the set of configurations that are reachable from $(q_0,\vecz)$
and from which the final state can be reached.
Observe that for any configuration $c \in C$ the language
$L(c) = \{ w \in (X \cup \bar X)^* \mid \exists \bu \colon c \xrightarrow{w} (q_f,\bu) \}$ is also offset-uniform
since $L(c) \subseteq \{ w \in (X \cup \bar X)^* \mid vw \in L(\calV)\}$ where $v \in (X \cup \bar X)^*$ is any word
with $(q_0,\vecz) \xrightarrow{v} c$.
Define the function $f\colon C\to \mathbb{Z}$ where $f(c)$ is the unique offset of the words in $L(c)$.
It remains to show that $|f(c)|$ is bounded double exponentially for all $c \in C$.

Let $M$ be the set of all configurations from which the final state can be reached (hence $C\subseteq M$).
Consider the following order on $\VASS$ configurations $Q \times \multiset{I}$: $(q, \bu) \leq (q', \bu')$ iff $q = q'$ and $\bu(i) \leq \bu'(i)$ for each $i \in I$. 
The cardinality of the set $\min(M)$ of minimal elements in $M$ with respect to this order is bounded doubly exponentially in the size of $\calV$. This follows directly from the fact that Rackoff's doubly-exponential bound~\cite{Rackoff78} on the length of a covering run does not depend on the start configuration (but only the size of the VASS and the final configuration). An explicit bound for $|\min(M)|$ is given in~\cite[Theorem~2]{BozzelliG11}.

Observe that if $c_1 \in M$ and $c_2 \in C$ with $c_1 \le c_2$ then $L(c_1) \subseteq L(c_2)$
and therefore $L(c_1)$ is also offset-uniform, having the same offset as $L(c_2)$.
Hence, if for two configurations $c_1,c_2\in C$ there exists a configuration $c\in M$ with $c \le c_1$
and $c \le c_2$, then $f(c_1)=f(c_2)$.
Since for every $c_2 \in C$ there exists $c_1 \in \min(M)$ with $c_1 \le c_2$,
the function $f$ can only assume doubly exponentially many values on $C$. 

Finally, we claim that $f(C) \subseteq \Z$ is an interval containing 0,
which proves that the norms of elements in $f(C)$ are bounded by the number of different values,
i.e., double exponentially.
Since we assumed $L(\calV) \neq \emptyset$,
some final configuration $(q_f,\bu) \in C$ is reachable from $(q_0,\vecz)$,
and therefore $0 \in f(C)$ since $\varepsilon \in L((q_f,\bu))$.
Consider the configuration graph $\mathcal{C}$ of $\calV$ restricted to $C$.
For any edge $c_1 \to c_2$ in $\mathcal{C}$ we have $|f(c_1) - f(c_2)| \le 1$
since VASS transitions consume at most one input symbol.
Moreover, the underlying undirected graph of $\mathcal{C}$ is connected since any configuration
is reachable from $(q_0,\vecz) \in C$.
Therefore $f(C)$ is an interval, which concludes the proof.
\end{proof}

Note that although $\Bounded(X,k)$ is a regular language for each $X$ and $k$,
\cref{thm:vass-boundedness} does not imply that every offset-uniform VASS language
is regular. For example, the $\VASS$ language $\{(x\bar{x})^m (y\bar{y})^n \mid m\ge n\}$
is offset-uniform, but it is not regular. This is because
\cref{thm:vass-boundedness} only implies boundedness of the number of occurrences of letters in the input
words, but the $\VASS$'s own counters might be unbounded.

The main consequence of \cref{thm:vass-boundedness} is that in a $\VASS$ we can track the offset
using a doubly succinct control state.
Thus, we have the following corollary.

\begin{corollary}
\label{cor:offset-zero}
The following problems can be decided in $\EXPSPACE$: 
Given a $\VASS$ or $\dsVASS$ $\calV$, does $\offset(w) = 0$ hold for all $w \in L(\calV)$? 
\end{corollary}

\begin{proof}
First assume $\calV$ is a $\VASS$.
We show that the problem can be reduced to the intersection non-emptiness problem for a $\VASS$ and a doubly succinct $\NFA$,
i.e., given a $\VASS$ $\calV$ and a doubly succinct NFA $\calA$, is the intersection $L(\calV) \cap L(\calA)$ nonempty?
One can construct in polynomial time a doubly succinct $\VASS$ for $L(\calV) \cap L(\calA)$,
as a product construction between $\calV$ and $\calA$.
Since the emptiness problem for $\dsVASS$ is in $\EXPSPACE$ (\cite[Theorem 5.1]{BaumannMTZ22}),
we can also decide emptiness of $L(\calV) \cap L(\calA)$ in $\EXPSPACE$.

Define the number $M = 2^{2^{p(|\calV|)}}$ where $p$ is the polynomial from \cref{thm:vass-boundedness}.
Let $K_0 = \{ w \in (X \cup \bar{X})^* \mid \offset(w) = 0 \}$.
According to \cref{thm:vass-boundedness}, we have $L(\calV)\subseteq K_0$
if and only if $L(\calV)\subseteq K_0 \cap \Bounded(X,M)$.
By the remarks above, it suffices to construct a doubly succinct NFA for the complement of $K_0 \cap \Bounded(X,M)$.
The following doubly succinct deterministic finite automaton $\calA$ recognizes $K_0 \cap \Bounded(X,M)$:
Given an input word over $X \cup \bar{X}$, the automaton tracks the current offset in the interval $[-M,M]$,
stored in the control state as a binary encoding of length $\log M = 2^{p(|\calV|)}$ together with a bit indicating the sign.
If the absolute value of the offset exceeds $M$, the automaton moves to a rejecting sink state.
The state representing offset $0$ is the initial and the only final state.
Finally, we complement $\calA$ to obtain a doubly succinct NFA $\bar \calA$, with a unique final state,
for the complement of $K_0 \cap \Bounded(X,M)$.

Now assume $\calV$ is a $\dsVASS$. Using Lipton's construction simulating doubly exponential counter values~\cite{lipton1976reachability}, we can construct
a (conventional) $\VASS$ $\calV'$, size polynomial in $|\calV|$, with the same language
(similar to \cite[Theorem 5.1]{BaumannMTZ22}).
We can now apply the above construction. 
\end{proof}

Next, we check for (DV) or (MV), assuming offset uniformity.
We will reduce both kinds of violations to the problem
of searching for \emph{marked Dyck factors}.
A word of the form $u \# v \bar \# w$ is called a \emph{marked Dyck factor}
if $u,v,w \in \{x, \bar x\}^*$ and $v \in \Dyck_x$.

Intuitively, if a (DV) occurs in a word $w$, there is a first time that the offset reaches $-1$.
Placing a $\bar{\#}$ at the place where this happens, and a $\#$ right at the beginning,
we have a word of the form $\#u\bar{\#}v$ where $u \in \Dyck_x$.
Similarly for (MV), we replace two letters $z \in X$ and $\bar{y} \in X$ with $z \neq y$
by $\#$ and $\bar{\#}$, respectively, and look for a word $u \# v \bar{\#} w$, where $v\in \Dyck_x$.

\begin{proposition}
\label{prop:VASScheck}
The following problems can be decided in $\EXPSPACE$: 
Given an offset-uniform $\VASS$ or $\dsVASS$ $\calV$, does $L(\calV)$ contain a marked Dyck factor? 
\end{proposition}

\begin{proof}
As in \cref{cor:offset-zero}, given a $\dsVASS$, we can convert to a polynomial-sized $\VASS$ with the same language and apply the following algorithm.

We again reduce to the intersection nonemptiness problem between a $\VASS$ and a doubly succinct $\NFA$, and use the fact that nonemptiness of $\dsVASS$ is in $\EXPSPACE$ 
\cite[Theorem 5.1]{BaumannMTZ22}.
As above, define the number $M = 2^{2^{p(|\calV|)}}$ where $p$ is the polynomial from \cref{thm:vass-boundedness}.
The automaton keeps track of the offset and also verifies that the input has the correct format $u \# v \bar \# w$ where $u,v,w \in \{x, \bar x\}^*$.
Furthermore, upon reaching $\#$ it starts tracking the current offset and verifies that (i)~the offset stays nonnegative,
(ii)~the offset never exceeds~$2M$, and (iii)~the offset is zero when reaching $\bar \#$.
If $L(\calV)$ intersects $L(\calA)$, then clearly $\calV$ is a positive instance of the problem.
Conversely, assume that $L(\calV)$ contains a word $u \# v \bar \# w$ with $v \in \Dyck_x$.
By offset-uniformity of $\calV$ and by~\cref{thm:vass-boundedness}, each prefix $v'$ of $v$ satisfies 
$\offset(v') = \offset(uv') - \offset(u) \le M - (-M) = 2M$.
Therefore $u \# v \bar \# w \in L(\calA)$.
\end{proof}

Let us put everything together.
Let $\rho \colon (X \cup \bar X)^* \to \{x, \bar x \}^*$ be the morphism
that replaces all letters from $X$ (resp., $\bar{X})$ by the letter $x$ (resp., $\bar x$).
Given a $\dsVASS$ $\calV$ over $X \cup \bar X$ we can construct in polynomial time
three $\dsVASS$ $\calV_\ioff, \calV_\idip, \calV_\imis$ where
\begin{equation*}
\label{eq:vasss}
\begin{aligned}
	L(\calV_\ioff) &= \rho(L(\calV)), \\
	L(\calV_\idip) &= \{ \# \rho(v) \bar \# \rho(\bar y w) \mid v \bar y w \in L(\calV) \text{ for some } v,w \in (X \cup \bar X)^*,\,y \in X \}, \\
	L(\calV_\imis) &= \{ \rho(u) \# \rho(v) \bar \# \rho(w) \mid u y v \bar z w \in L(\calV) \text{ for some } u,v,w \in (X \cup \bar X)^*, \, y \neq z \in X \}.
\end{aligned}
\end{equation*}
Observe that $L(\calV) \subseteq \Dyck_x$ if and only if $L(\calV_\ioff)$ has uniform offset $0$
and $L(\calV_\idip)$ and $L(\calV_\imis)$ do not contain marked Dyck factors.

Hence, to decide whether $L(\calV) \subseteq \Dyck_X$ we first test that $L(\calV_\ioff)$ has uniform offset 0,
using~\cref{cor:offset-zero}, rejecting if not.
Otherwise, we can apply \cref{prop:VASScheck} to test whether
$L(\calV_\idip)$ or $L(\calV_\imis)$ contain marked Dyck factors.
If one of the tests is positive, we know $L(\calV) \not \subseteq \Dyck_X$, otherwise $L(\calV) \subseteq \Dyck_X$.

\begin{theorem}
\label{thm:VASScoverInD}
Given a $\dsVASS$ $\calV$ over the alphabet $X \cup \bar X$, 
checking whether $L(\calV) \subseteq \Dyck_X$ is $\EXPSPACE$-complete.
\end{theorem}

Let us remark that \cref{thm:VASScoverInD} can also be phrased slightly more
generally.  Above, we have defined the language of a $\VASS$ to be the set of
input words for which a final state is reached. Such languages are also called
\emph{coverability languages}. Another well-studied notion is the
\emph{reachability language} of a $\VASS$, which consists of those words for
which a configuration $(q_f,\vecz)$ is reached. Moreover, a $\VASS$ is
\emph{deterministic} if for each input letter $x$ and each state $q$, there is
at most one $x$-labeled transition starting in $q$ (and there are no
$\varepsilon$-transitions). We can now phrase \cref{thm:VASScoverInD} as
follows: Given a $\VASS$ coverability language $K$ and a reachability language
$L$ of a deterministic $\VASS$, it is $\EXPSPACE$-complete to decide whether
$K\subseteq L$. This is in contrast to inclusion problems where $K$ is drawn from
a subclass of the coverability languages: This quickly leads to
Ackermann-completeness~\cite{CzerwinskiH22}. In fact, even if we replace
$\Dyck_X$ in \cref{thm:VASScoverInD} with the set of prefixes of
$\Dyck_{\{x\}}$, the problem becomes Ackermann-complete (see the full version
of this work).

\section{Checking Dyck Inclusion for Programs} %
\label{sec:algorithm}

We now describe our algorithm for checking inclusion in $\Dyck_X$ for programs.
Our argument is similar to the case of $\dsVASS$: we first construct three auxiliary programs
$\asyncp_\ioff$, $\asyncp_\idip$, and $\asyncp_\imis$,
and then we use them to detect each type of violation in the original program.
We construct the program $\asyncp_\ioff$ for checking offset violation by projecting the Dyck letters
to the one-dimensional Dyck alphabet $\{x,\bar{x}\}$.
The programs $\asyncp_\idip$ and $\asyncp_\imis$ are constructed by first placing two markers like for $\VASS$,
and then projecting to $\{x,\bar{x}\}$.

As in the algorithm for $\VASS$, we check whether $L(\asyncp_\ioff)$ has uniform offset 0,
and whether $L(\asyncp_\idip)$ and $L(\asyncp_\imis)$ contain marked Dyck factors.
For these checks, we convert the three programs into $\dsVASS$ $\calV_\ioff$, $\calV_\idip$, and $\calV_\imis$, respectively,
in such a way that violations are preserved.
To be more precise, this conversion from programs to $\dsVASS$ will preserve the \emph{downward closure}
with respect to a specific order that we define below.
The global downward closure procedure is obtained by composing a local downward closure procedure applied to each task.
On the task level, the order $\extsw$ is a combination of the subword order on the handler names in $\Gamma$
and the syntactic order of $\Dyck_X$ over the event letters.
The core technical result is a transformation from context-free grammars into $\dsNFA$
which preserve the downward closure with respect to $\extsw$.

One key aspect of our downward closure construction is an important condition on the pumps that appear
in the context-free grammar. 
\begin{definition}
A context-free grammar $\calG$ is \emph{tame-pumping} if for every pump $A \derivs u A v$, we have $\offset(u)\ge 0$ and $\offset(v)=-\offset(u)$.
A derivation $A \derivs u A v$ is called an \emph{increasing pump} if $\offset(u)>0$, otherwise it is called a \emph{zero pump}.
An asynchronous program is \emph{tame-pumping} if its grammar is tame-pumping.
\end{definition}
Note that while our definition of a tame-pumping grammar is syntactic, it actually only depends on the generated language, assuming every nonterminal occurs in a derivation: In that case, a grammar is tame-pumping if and only if (i)~the set of offsets and (ii)~the set of dips of words in its language are both finite.

The following lemma summarizes some properties of tame-pumping and why it is useful for our algorithm.
The proof can be found in the full version.

\begin{restatable}{lemma}{checktamepumping}
\label{lem:checkTamePumping}
\begin{enumerate}
\item
	We can check in $\coNP$ whether a given context-free grammar over $\{x,\bar x\}$ is tame-pumping.
	Furthermore, given a nonterminal $A_0$, we can check in $\NP$ whether $A_0$ has a zero pump (resp., increasing pump).
\item \label{lem:dipBoundedCFG}
	There exists a polynomial $p$ such that, if $\calG$ is tame-pumping, then for every nonterminal $A$ of $\cG$ and every $w \in L(\cG,A)$ we have $\dip(w) \le 2^{p(|\calG|)}$.
\item
If $\asyncp$ is not tame-pumping, then $L(\asyncp) \not \subseteq \Dyck_X$.
\end{enumerate}
\end{restatable}

Thus, if $\asyncp$ is not tame-pumping, the refinement checking algorithm rejects immediately.
From now on, we assume that $\asyncp$ is tame-pumping.

\subsection{Combining the subword order and the syntactic order}
\label{sec:dc-def}

Suppose $\Gamma$ is an alphabet and let $\Theta=\Gamma\cup\{x,\bar{x}\}$. Define $\bar{a}=a$ for $a \in \Gamma$.
By $\subword$, we denote the \emph{subword ordering} on $\Gamma^*$,
i.e.\ $u \subword v$ if and only if $u$ can be obtained from $v$ by deleting some letters.
Formally there exist words $u_1, \dots, u_n, v_0, \dots, v_n \in \Gamma^*$
such that $u = u_1 \cdots u_n$ and $v = v_0 u_1 v_1 \cdots u_n v_n$.
For $u,v\in\{x,\bar{x}\}^*$, we write $u\sdyck v$ if $\offset(u)=\offset(v)$ and $\dip(u)\ge \dip(v)$.
In fact, $\sdyck$ is the \emph{syntactic order} with respect to the Dyck language,
i.e.\ if $u \sdyck v$ and $rus\in \Dyck_x$ then $rvs\in \Dyck_x$ for all $r,s$.
We define the ordering $\extsw'$ on $\Theta^*$ by
$z_1 \extsw' z_2$ if and only if $\pi_{x,\bar{x}}(z_1)\sdyck \pi_{x,\bar{x}}(z_2)$,
and $\pi_\Gamma(z_1)\subword \pi_\Gamma(z_2)$.
For example, $a \bar{x}x c \extsw' x abc \bar{x}$ because $ac$ is a subword of $abc$,
and both $\bar{x} x$ and $x \bar{x}$ have offset 0, but $\bar{x} x$ has a larger dip.

Let $\#, \bar \#$ be two fresh letters, called \emph{markers}.
The set of \emph{marked words} is defined as
\[ \Decomps=\Theta^*\{\varepsilon,\#\}\Theta^*\{\varepsilon,\bar{\#}\}\Theta^*. \]
A marked word should be viewed as an infix of a larger word $u \# v \bar \# w$.
The set of \emph{admissible} marked words, denoted by $\AdmDecomps$,
consists of those words $z\in\Decomps$ which are an infix of a word $u \# v \bar \# w$ where $v \in \Dyck_x$.
For example, a marked word $u \# v$ is admissible if $v$ is a prefix of a Dyck word.

On the set of admissible marked words, we define an ordering $\extsw$.
To do so, we first define for each marked word $z \in \Decomps$
two words $\inside{z}$ and $\outside{z}$ in $\Theta^*$ as follows:
Let $u,v,w \in \Theta^*$ such that either $z = v$, $z = u \# v$, $z = v \bar \# w$, or $z = u \# v \bar \# w$.
Then we define $\inside{z} = v$ and $\outside{z} = uw$ (here, $u = \varepsilon$ if it is not part of $z$, same for $w$).
Given two admissible marked words $z_1, z_2 \in \AdmDecomps$
we define $w \extsw w'$ if and only if $z_1$ and $z_2$ contain the same markers,
and $\inside{z_1} \extsw' \inside{z_2}$, and $\outside{z_1} \extsw' \outside{z_2}$.
For example, $a \bar{x}x c \# a \extsw x abc \bar{x} \# ab$ because $a \bar{x}x c \extsw' x abc \bar{x}$
and $a \extsw' ab$.

For a language $L\subseteq \Decomps$ we denote by $\dc{L}$ the
downward closure of $L$ within $\AdmDecomps$ with respect to the ordering
$\extsw$. Thus, we define:
\[ \dc{L}=\{u\in \AdmDecomps \mid \exists v\in L \cap \AdmDecomps \colon u\extsw v\}. \]

\begin{theorem}\label{thm:extended-downclosure}
	Given a tame-pumping $\CFG$ $\calG$, we can compute in polynomial space a doubly
	succinct $\NFA$ $\calA$ such that
	$\dc{L(\calA)}=\dc{L(\calG)}$ and $|\calA|$ is polynomially bounded in $|\calG|$.
\end{theorem}

We explain how to prove \cref{thm:extended-downclosure} in \cref{sec:extended-downclosure}.
Let us make a few remarks.
While downward closed sets with respect to the subword ordering are always regular,
this does not hold for $\extsw$.
Consider the language $L = (ax)^*$ where $a \in \Gamma$ is a handler name and $x \in X$ is an event letter.
Then $\dc{L}$ consists of all words $w \in \{a,x,\bar x\}^*$ where $|w|_a \le |w|_x - |w|_{\bar x}$,
which is not a regular language.
Furthermore, the automaton in \cref{thm:extended-downclosure} may indeed require
double exponentially many states. For example, given a number $n$,
consider the language $L=\{u\bar{x}^{2^n}\#x^{2^n}\bar{u} \mid
u\in\{ax,bx\}^*\}$ where $\Gamma = \{a,b\}$ is the set of handler names
and $X = \{x\}$. Here we define $\overline{a_1 a_2 \cdots a_n} = \bar a_n \cdots \bar a_2 \bar a_1$
for a word $a_1 \cdots a_n \in \{a,b,x\}^*$ where $\bar a = a$ and $\bar b = b$.
This is generated by a tame-pumping context-free
grammar of size linear in $n$. However, for any $\calA$ with
$\dc{L(\calA)}=\dc{L}$, projecting to just $a$ and $b$ yields the
language $K=\dc{\{uu^\rev \mid u\in\{a,b\}^*,~|u|\le 2^n\}}$, for which
an NFA requires at least $2^{2^n}$ states.

Finally, note that the restriction to admissible words is crucial: If we defined the ordering $\extsw$ on all words of $\Decomps$, then for the tame-pumping language $L=\{x^n\#\bar{x}^n \mid n\in\N\}$, the downward closure would not be regular, because an NFA would be unable to preserve the unbounded offset at the separator $\#$. A key observation in this work is that in combination with tame pumping, admissibility guarantees that the offset at the borders $\#$ and $\bar{\#}$ is bounded (see~\cref{lem:boundingOffsetInUPFtree}), which enables a finite automaton to preserve it.

Given a tame-pumping asynchronous program $\asyncp$, we can now compute a $\dsVASS$~$\calV$ with the same downward closure:
Its counters are the handler names $a \in \Gamma$ in $\asyncp$.
For each nonterminal $A$ we apply \cref{thm:extended-downclosure} to $\calG_A$,
which is the grammar of $\asyncp$ with start symbol $A$, and obtain a $\dsNFA$ $\calB_A$.
We replace each transition $q \xhookrightarrow{a,A} q'$ by the following gadget:
First, it decrements the counter for the handler name $a$.
Next, the gadget simulates the $\dsNFA$ $\calB_A$ where handlers $b \in \Gamma$ are interpreted as counter increments.
Finally, when reaching the final state of $\calB_A$ we can non-deterministically switch to $q'$.

\begin{restatable}{corollary}{aptovass}
	\label{cor:ap-to-vass}
	Given an asynchronous program $\asyncp$ with tame-pumping,
	we can compute in polynomial space a doubly succinct $\VASS$ $\calV$
	such that $\dc{L(\asyncp)}=\dc{L(\calV)}$ and $|\calV|$ is polynomially bounded in $|\asyncp|$.
\end{restatable}
The details of the proof are given in the full version.

\subsection{The algorithm}
\label{subsec:algorithm}

We are now ready to explain the whole algorithm.
Given an asynchronous program $\asyncp = (Q, X \cup \bar X, \Gamma, \cG, \Delta, q_0, q_f, \gamma_0)$,
we want to check if $L(\asyncp) \subseteq \Dyck_X$.
Recall that, wlog, we can assume all nonterminals are useful,
meaning every nonterminal is involved in some accepting run.
The algorithm is presented in~\cref{alg:outline}.
As a first step, the algorithm verifies that $\asyncp$ is tame-pumping using \cref{lem:checkTamePumping}.
Next we construct the following auxiliary asynchronous programs $\asyncp_\ioff$, $\asyncp_\idip$, $\asyncp_\imis$,
to detect offset, dip, and mismatch violations in $L(\asyncp)$.
Let $\rho \colon (X \cup \bar X)^* \to \{x,\bar x\}^*$ be the morphism 
which replaces all letters in $X$ by unique letter $x$ and all letters in $\bar X$ by unique letter $\bar x$.
The programs $\asyncp_\ioff$, $\asyncp_\idip$, $\asyncp_\imis$ recognize the following languages
over the alphabet $\{ x, \bar x, \#, \bar \# \}$:
\begin{equation}
\label{eq:aux-programs}
\begin{aligned}
	L(\asyncp_\ioff) &= \{ \rho(w) \mid w \in L(\asyncp) \}, \\
	L(\asyncp_\idip) &= \{ \# \rho(v) \bar \# \rho(\bar y w) \mid v \bar y w \in L(\asyncp) \text{ for some } v,w \in (X \cup \bar X)^*,\,y \in X \}, \\
	L(\asyncp_\imis) &= \{ \rho(u) \# \rho(v) \bar \# \rho(w) \mid u y v \bar z w \in L(\asyncp), \\ & \qquad\qquad\qquad\qquad\qquad\qquad \text{for some } u,v,w \in (X \cup \bar X)^*, \, y \neq z \in X \}.
\end{aligned}
\end{equation}
In fact, if the original asynchronous program $\asyncp$ is tame-pumping,
we can ensure that $\asyncp_\ioff$, $\asyncp_\idip$, $\asyncp_\imis$ are also tame-pumping
(see the full version for details).

\begin{algorithm}[t]
  \caption{Checking non-inclusion of $L(\asyncp)$ in the Dyck language $\Dyck_X$ in $\EXPSPACE$. \label{alg:outline}}
  \SetKw{Return}{return}
  \SetKwInOut{Input}{input}
  \SetKwSwitch{Guess}{Violation}{Other}{guess}{do}{violation}{otherwise}{end}
  \Input{Asynchronous program $\asyncp$ for a language $L \subseteq (X\cup \bar{X})^*$} \\
  \lIf{\text{$\asyncp$ does not have tame-pumping (\emph{\cref{lem:checkTamePumping}})}}{\Return{$L \not\subseteq \Dyck_X$}}
  Construct asynchronous programs $\asyncp_\ioff$, $\asyncp_\idip$, $\asyncp_\imis$ (\cref{eq:aux-programs}).
  
  Construct $\dsVASS$ $\calV_\ioff$, $\calV_\idip$, $\calV_\imis$
  with $\dc{L(\calV_{\mathsf{x}})} = \dc{L(\asyncp_{\mathsf{x}})}$ for $\mathsf{x} \in \{\ioff, \idip, \imis\}$ (\cref{cor:ap-to-vass}).
  
  \lIf{\text{$\calV_\ioff$ does not have uniform offset 0 (\emph{\cref{cor:offset-zero}})}}{\Return{$L \not\subseteq \Dyck_X$}}
  \lIf{\text{$L(\calV_\idip)$ or $L(\calV_\imis)$ contains a marked Dyck factor (\emph{\cref{prop:VASScheck}})}}{\Return{$L \not\subseteq \Dyck_X$}}
  
  \Return{$L \subseteq \Dyck_X$}
\end{algorithm}
 
It remains to verify whether $L(\asyncp_\ioff)$ has uniform offset 0,
	and $L(\asyncp_\idip)$ and $L(\asyncp_\imis)$ do not contain marked Dyck factors.
By \cref{cor:ap-to-vass} we can compute for each $\mathsf{x} \in \{\ioff, \idip, \imis\}$ 
a $\dsVASS$ $\calV_\mathsf{x}$ with $\dc{L(\calV_{\mathsf{x}})} = \dc{L(\asyncp_{\mathsf{x}})}$.
Since $\extsw$ preserves offsets we know that $L(\asyncp_\ioff)$ has uniform offset 0
if and only if $L(\calV_\ioff)$ has uniform offset 0, which can be decided in exponential space by~\cref{cor:offset-zero}.
Finally, we check whether $L(\calV_\idip)$ or $L(\calV_\imis)$ contain a marked Dyck factor by~\cref{prop:VASScheck}.
This is correct, because a language $L$ contains a marked Dyck factor if and only if $\dc{L}$ contains a marked Dyck factor:
On the one hand, the ``only if'' direction is clear because $L \subseteq \dc{L}$.
On the other hand, if $u \# v \bar \# w \in \dc{L}$ is a marked Dyck word
then there exists a word $u' \# v' \bar \# w' \in L$ with $v \sdyck v'$, and therefore $v' \in \Dyck_x$.

\newcommand{\shuf}{\mathsf{shuf}}
\newcommand{\sNFA}{\mathsf{sNFA}}
\newcommand{\extswlr}{\sqsubseteq_{\markl,\markr}}
\newcommand{\extswl}{\sqsubseteq_{\markl}}
\newcommand{\extswr}{\sqsubseteq_{\markr}}

\section{Computing Downward Closures and the Proof of Theorem~\ref{thm:extended-downclosure}}
\label{sec:extended-downclosure}
\begin{figure}

\centering

\begin{tikzpicture}[yscale=-1]

\tikzstyle{xkcd} = [] %
\tikzstyle{blob}=[inner sep = 1pt, fill, circle]
\tikzstyle{undivided}=[draw=blue!70!black,fill=blue, fill opacity=0.3]
\tikzstyle{undividednode}=[inner sep = 1pt, fill=blue!70!black, circle, fill opacity = 1]

\colorlet{ourGray}{gray} %

\begin{scope}
\coordinate (a) at (0,0);
\coordinate (b) at (-3,5);
\coordinate (c) at (3,5);

\draw[ourGray,xkcd] (1.5,3.5) ++(-0.4,0.9)  ++(0.35,0) ++(0.1,-0.1) -- (1.55,5);

\draw[ourGray,xkcd] (0,4) ++(-0.4,0.9)  ++(0.35,0) ++(0.1,-0.1) -- (0.05,5);

\draw (a) -- (b) -- (c) -- cycle;

\coordinate (d) at (0,1.3);
\coordinate (e) at (0,2);
\node[blob] (f) at (0,3) {};
\node[blob, label={[yshift=-2em]:\footnotesize $\#$}] (m1) at (-1,5) {};
\node[blob, label={[yshift=-2em]:\footnotesize $\overline \#$}] (m2) at (1,5) {};

\draw[ourGray,xkcd] (0,4) -- ++(-0.33,-0.33);
\draw[ourGray,xkcd] (1.5,3.5) -- (0,2.4);

\draw[semithick,xkcd] (a.center) -- (d.center)  -- (e.center)  -- (f.center);
\draw[semithick,xkcd] (f.center) -- (m1.center);
\draw[semithick,xkcd] (f.center) -- (m2.center);

\filldraw[undivided] (d.center) -- ++(-0.7,0.9)  -- ++(0.6,0) -- (e.center) -- ++(0.15,0.2) -- ++(0.6,0) -- cycle;

\node[undividednode] at (a) {};
\node[undividednode] at (d) {};
\node[undividednode] at (e) {};

\filldraw[undivided] (1.5,3.5) node[undividednode] {} -- ++(-0.4,0.9)  -- ++(0.35,0) -- ++(0.1,-0.1) node[undividednode] {} -- ++(0.1,0.1) -- ++(0.35,0) -- cycle;

\filldraw[undivided] (0,4) node[undividednode] {} -- ++(-0.4,0.9)  -- ++(0.35,0) -- ++(0.1,-0.1) node[undividednode] {} -- ++(0.1,0.1) -- ++(0.35,0) -- cycle;

\node at (3,2.3) {$\implies$};
\end{scope}

\tikzstyle{divided}=[draw=red!70!black,fill=red, fill opacity=0.3]
\tikzstyle{dividednode}=[inner sep = 1pt, fill=red!70!black, circle, fill opacity = 1]

\begin{scope}[xshift=15em,yshift=-2em]

\coordinate (a) at (0,1);
\coordinate (a1) at (0,1.25);
\coordinate (a2) at (0,1.5);
\coordinate (a3) at (0,1.75);
\coordinate (b) at (0,2);
\coordinate (c) at (-1,5);
\coordinate (d) at (1,5);

\draw[black!60,arrows={- ____ Straight Barb[left,reversed] ____ Straight Barb[left,reversed] ____ Straight Barb[left,reversed] ____ Straight Barb[left,reversed]}] (a.center) -- (b.center);
\draw[black!60,arrows={- ____ Straight Barb[left] ____ Straight Barb[left] ____ Straight Barb[left] ____ Straight Barb[left]}] (b.center) -- (a.center);

\draw[semithick] (a.center) -- (b.center);

\draw[black!60,arrows={- ____ Straight Barb[left,reversed] ____ Straight Barb[left,reversed] ____ Straight Barb[left,reversed] ____ Straight Barb[left,reversed] ____ Straight Barb[left,reversed] ____ Straight Barb[left,reversed] ____ Straight Barb[left,reversed] ____ Straight Barb[left,reversed] ____ Straight Barb[left,reversed] ____ Straight Barb[left,reversed] ____ Straight Barb[left,reversed] ____ Straight Barb[left,reversed] ____ Straight Barb[left,reversed]}] (b.center) -- (c.center);
\draw[black!60,arrows={- ____ Straight Barb[left] ____ Straight Barb[left] ____ Straight Barb[left] ____ Straight Barb[left] ____ Straight Barb[left] ____ Straight Barb[left] ____ Straight Barb[left] ____ Straight Barb[left] ____ Straight Barb[left] ____ Straight Barb[left] ____ Straight Barb[left] ____ Straight Barb[left] ____ Straight Barb[left]}] (c.center) -- (b.center);

\draw[semithick] (b.center) -- (c.center);

\draw[black!60,arrows={- ____ Straight Barb[left,reversed] ____ Straight Barb[left,reversed] ____ Straight Barb[left,reversed] ____ Straight Barb[left,reversed] ____ Straight Barb[left,reversed] ____ Straight Barb[left,reversed] ____ Straight Barb[left,reversed] ____ Straight Barb[left,reversed] ____ Straight Barb[left,reversed] ____ Straight Barb[left,reversed] ____ Straight Barb[left,reversed] ____ Straight Barb[left,reversed] ____ Straight Barb[left,reversed]}] (b.center) -- (d.center);
\draw[black!60,arrows={- ____ Straight Barb[left] ____ Straight Barb[left] ____ Straight Barb[left] ____ Straight Barb[left] ____ Straight Barb[left] ____ Straight Barb[left] ____ Straight Barb[left] ____ Straight Barb[left] ____ Straight Barb[left] ____ Straight Barb[left] ____ Straight Barb[left] ____ Straight Barb[left] ____ Straight Barb[left]}] (d.center) -- (b.center);

\draw[semithick] (b.center) -- (d.center);

\node[blob] at (a) {};
\node[blob] at (b) {};
\node[blob, label={[yshift=-2em]:\footnotesize $\#$}] at (c) {};
\node[blob, label={[yshift=-2em]:\footnotesize $\overline \#$}] at (d) {};

\filldraw[divided] (-0.33,3) node[dividednode] {} -- ++(-0.2,0.1) -- (-0.86,4.1) -- (-0.66,4) node[dividednode] {}%
-- (-0.52,4.2) -- (-0.2,3.2) -- cycle;

\node at (1.5,3) {$\implies$};

\end{scope}

\begin{scope}[xshift=25em,yshift=-5em]

\coordinate (a) at (0,3);
\coordinate (b) at (0,4);
\coordinate (c) at (-0.5,4.8);
\coordinate (d) at (0.5,4.8);

\draw[black!60,arrows={- ____ Straight Barb[left,reversed] ____ Straight Barb[left,reversed] ____ Straight Barb[left,reversed] ____ Straight Barb[left,reversed]}] (a.center) -- (b.center);
\draw[black!60,arrows={- ____ Straight Barb[left] ____ Straight Barb[left] ____ Straight Barb[left] ____ Straight Barb[left]}] (b.center) -- (a.center);

\draw[semithick] (a.center) -- (b.center);

\draw[black!60,arrows={- ____ Straight Barb[left,reversed] ____ Straight Barb[left,reversed] ____ Straight Barb[left,reversed] ____ Straight Barb[left,reversed]}] (b.center) -- (c.center);
\draw[black!60,arrows={- ____ Straight Barb[left] ____ Straight Barb[left] ____ Straight Barb[left] ____ Straight Barb[left]}] (c.center) -- (b.center);

\draw[semithick] (b.center) -- (c.center);

\draw[black!60,arrows={- ____ Straight Barb[left,reversed] ____ Straight Barb[left,reversed] ____ Straight Barb[left,reversed] ____ Straight Barb[left,reversed]}] (b.center) -- (d.center);
\draw[black!60,arrows={- ____ Straight Barb[left] ____ Straight Barb[left] ____ Straight Barb[left] ____ Straight Barb[left]}] (d.center) -- (b.center);

\draw[semithick] (b.center)  -- (d.center);

\node[blob] at (a) {};
\node[blob] at (b) {};
\node[blob] at (c) {};
\node[blob] at (d) {};
\node[blob, label={[yshift=-2em]:\footnotesize $\#$}] at (c) {};
\node[blob, label={[yshift=-2em]:\footnotesize $\overline \#$}] at (d) {};

\end{scope}

\end{tikzpicture}

\caption{Abstracting undivided pumps (in blue) and divided pumps (in red).}
\label{fig:abstracting-pumps}

\end{figure}
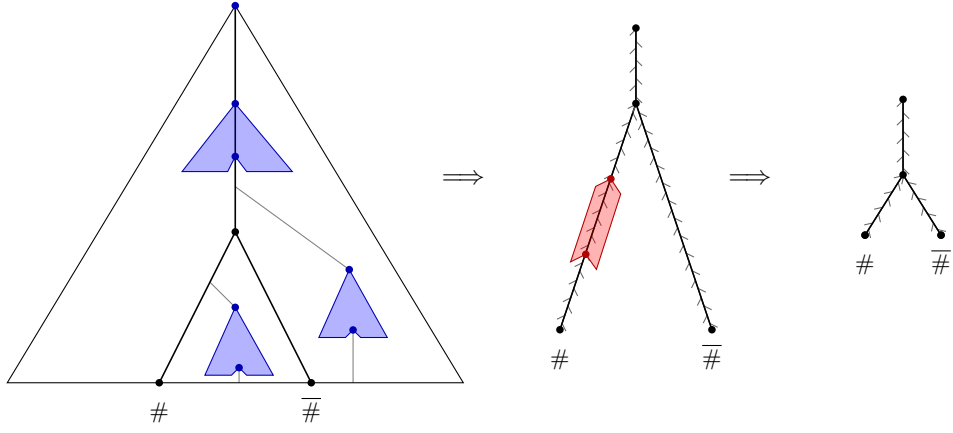

It remains to show how the automaton $\calA$ for the downward closure in \cref{thm:extended-downclosure} is constructed. 
As a warm-up, let us illustrate how to construct from a context-free grammar~$\cG$ an NFA $\calA$
for the subword closure of $L(\cG)$, cf.\ \cite{Courcelle1991}.
Here, \emph{subword closure} refers to the downward closure with respect to the subword ordering $\preccurlyeq$.
Notice that this is a special case of \cref{thm:extended-downclosure}, namely where $L(\cG) \subseteq \Gamma^*$.
The basic idea is that every derivation tree of $\cG$ can be obtained by inserting pumps
into a \emph{skeleton}---a derivation tree without vertical repetitions of nonterminals.
The skeleton can be guessed by an (exponentially large) automaton $\calA$
and the effects of pumps are abstracted as follows:
For each nonterminal $A$ one can compute the subalphabets $\Gamma_{A,\markl}, \Gamma_{A,\markr} \subseteq \Gamma$
containing all letters occurring on the left side $u$ and the right side $v$ of a pump $A \derivs uAv$.
Instead of inserting pumps, the automaton for the subword closure inserts arbitrary words
$u' \in \Gamma_{A,\markl}^*$ and $v' \in \Gamma_{A,\markr}^*$ on the left or right side of $A$, respectively.
This is sufficient because for any word $w$, the subword closure of the language $w^*$ contains
exactly those words that consist only of letters present in $w$.

The difficulty in proving \cref{thm:extended-downclosure} is to preserve, not only the subword closure, 
but also the downward closure with respect to the syntactic order $\sdyck$ on the letters $\{x,\bar x\}$.
To do so, we need to distinguish between two types of pumps.
Consider the derivation tree for a marked word $z = u \# v \bar \# w$, depicted left in \cref{fig:abstracting-pumps}.
Observe that removing one of the three pumps in blue does not change the offset of $\inside{z} = v$ or $\outside{z} = uw$,
because $\cG$ is tame-pumping. Such pumps, which are completely contained in $\inside{z}$ or $\outside{z}$, will be called \emph{undivided}.
However, one needs to be more careful when removing \emph{divided} pumps, e.g., the red pump in the second derivation tree of \cref{fig:abstracting-pumps}.
Removing the red pump decreases the offset of $\outside{z}$, while increasing
the offset of $\inside{z}$ by the same amount.

We will proceed in two transformations, which preserve the downward closure w.r.t.\ $\extsw$.
In the first transformation we obtain a grammar whose derivation trees do not contain any undivided pumps.
In the second step we additionally eliminate divided pumps.

\subsection{Abstracting undivided pumps}
\label{ssec:absOrdPumps}

Recall that $\Decomps = \Theta^* \{\#,\varepsilon\} \Theta^* \{\bar \#,\varepsilon\} \Theta^*$
where $\Theta = \Gamma \cup \{x,\bar x\}$.
In the following we only consider \emph{uniformly marked} grammars $\cG$,
that is, we assume $L(\cG)$ is contained in one of the subsets $\Theta^* \# \Theta^* \bar \# \Theta^*$, $\Theta^* \# \Theta^*$, $\Theta^* \bar \# \Theta^*$, or $\Theta^*$.
This is not a restriction since we can split the given grammar $\cG$ into four individual grammars,
covering the four types of marked words, and treat them separately.
This allows us to partition the set of nonterminals $N$ into $N_{\# \bar \#} \cup N_\# \cup N_{\bar \#} \cup N_0$
where $N_{\# \bar \#}$-nonterminals only produce marked words in $\Theta^* \# \Theta^* \bar \# \Theta^*$,
$N_\#$-nonterminals only produce marked words in $\Theta^* \# \Theta^*$, etc.
A pump $A \derivs uAv$ is \emph{undivided} if $A \in N_{\# \bar \#} \cup N_0$, and \emph{divided} otherwise.
Our first goal will be to eliminate undivided pumps.
A derivation tree without undivided pumps may still contain exponentially large subtrees below $N_0$-nonterminals.
Such subtrees will also be ``flattened'' in this step, see the first transformation step in \cref{fig:abstracting-pumps}.

\begin{definition}
	\label{defn:UBFtree}
	A context-free grammar $\cG = (N,\Theta \cup \{\#,\bar \#\},P,S)$ is \emph{almost-pumpfree} iff
  \begin{description}
    \item[(C1)] $\cG$ does not have undivided pumps, and
    \item[(C2)]\label{it:specialPathC2} for all productions $A \to \alpha$ with $A \in N_0$ either $\alpha = a \in \Theta$
    or $\alpha = (\Gamma')^*$ for some $\Gamma' \subseteq \Gamma$.
  \end{description} 
\end{definition}

We will now explain how to turn any uniformly marked $\CFG$ into an almost-pumpfree one.
The resulting (extended) grammar will be exponentially large but can be represented succinctly.
Recall that a \emph{succinct ECFG} ($\sECFG$) is an extended context-free grammar $\cG$
whose nonterminals are polynomially long strings
and whose productions are given by finite-state transducers.
For example, one of the transducers accepts the finite relation of all triples $(A,B,C)$ such that there exists a production $A \to BC$.
Productions either adhere to Chomsky normal form or have the form $A \to B$.
The latter enables us to simulate $\PSPACE$-computations in the grammar without side effects,
see \cref{ECFG-PSPACE} below.

\begin{restatable}{proposition}{converttosecfg}\label{lem:convertTosECFG}
	Given a uniformly marked tame-pumping $\CFG$ $\cG$,
	one can compute in polynomial space a tame-pumping almost-pumpfree $\sECFG$ $\cG'$ such that $\dc{L(\cG)}=\dc{L(\cG')}$
and $|\cG'|$ is polynomially bounded in $|\cG|$.
\end{restatable}

To prove \cref{lem:convertTosECFG}, we first need some auxiliary results,
which are mainly concerned with computing the minimal dips and letter occurrences
within undivided pumps of a grammar~$\cG$.
Recall that for the subword closure we computed for each nonterminal $A$ the subalphabets $\Gamma_{\markl,A}$ and $\Gamma_{\markr,A}$,
and inserted arbitrary words over $\Gamma_{\markl,A}$ and $\Gamma_{\markr,A}$ left and right to the nonterminal $A$.
For the refined order $\extsw$ we may only use a letter $a \in \Gamma$
after simulating the minimal dip which is required to produce the letter $a$.

For a word $w \in \Theta^*$ we define the set $\psi(w)$ of all pairs $(n,m) \in \N^2$
such that $n \ge \dip(w)$ and $m = n + \offset(w)$.
In other words, $\psi(w)$ is the reachability relation induced by $w$, interpreted as counter instructions.
Recall that \emph{Presburger arithmetic} is the first-order theory of $(\N,+,<,0,1)$.
As an auxiliary step, we will compute existential Presburger formulas capturing the relation
$\psi(u) \times \psi(v)$ for all pumps $A \derivs uAv$ of a nonterminal $A$.

In the following lemma, when we say that we can \emph{compute a formula for a relation $R \subseteq \N^k$ in polynomial space}, we mean that there is non-deterministic polynomial-space algorithm,
where each non-deterministic branch computes a polynomial-size formula for a relation $R_i$
such that if  $R_1,\ldots,R_n$ are the relations of all the branches,
then $R = \bigcup_{i = 1}^n R_i$.
Here we tacitly use the fact that $\NPSPACE = \PSPACE$ \cite{savitch1970relationships}.

\begin{lemma}
	\label{lem:PAformulaEffect}
	Given an offset-uniform $\CFG$ with $L(\cG) \subseteq \Theta^* \$ \Theta^*$,
	where $\$ \notin \Theta$,
	we can compute in polynomial space an existential Presburger formula for the relation
	\[
		\bigcup_{u \$ v \in L} \psi(u) \times \psi(v) \subseteq \N^4.
	\]
\end{lemma}
\begin{proof}[Proof sketch]
The result of \cref{lem:PAformulaEffect} was already proved in \cite[Proposition~3.8]{BaumannGMTZ23},
under the additional assumption that the given context-free grammar $\calG$ for $L$ is \emph{annotated}
(they even show that in this case the formula can be computed in $\NP$).
We call $\calG$ annotated if for every nonterminal $A$ the minimal dip that can be achieved by a word in 
$L(\cG,A)$ is given as an input, denoted by $\mindip(A)$.
Hence, it remains to show how to compute the annotation of an offset-uniform grammar in $\PSPACE$,
which is possible using a simple saturation algorithm.
For each nonterminal $A$, the algorithm stores a number $D(A)$ satisfying $D(A) \ge \mindip(A)$.
Initially, $D(A)$ is set to an upper bound for $\mindip(A)$,
which by \cref{lem:checkTamePumping}~(2) can be chosen to be exponentially large in $|\calG|$.
In each round the function $D$ is updated as follows:
For each production $A \to BC$ we set $D(A)$ to the minimum of $D(A)$ and $\max\{ D(B), D(C) - \offset(B) \}$,
where $\offset(B)$ is the unique offset of $L(\cG,B)$.
Clearly, the algorithm can be implemented in polynomial space since the numbers are bounded exponentially.
Termination of the algorithm is guaranteed since the numbers $D(A)$ are non-increasing.
\end{proof}

With \cref{lem:PAformulaEffect} in hand, we can now prove the following lemma,
which allows us to check whether pumps with certain letter occurrences exist
for certain minimal dips.

\begin{restatable}{lemma}{dipsforspawns}\label{dips-for-spawns}
	Given a tame-pumping $\CFG$ $\cG$ such that $L(\cG) \subseteq \Decomps$,
	a nonterminal~$A$ in $\cG$, a letter $a \in \Gamma$ and two numbers $d_\markl,d_\markr \in \N$,
	we can decide in $\PSPACE$ if there exists a derivation $A \derivs uAv$
	such that $u$ contains the letter $a$ (or symmetrically, whether $v$ contains the letter $a$), $\dip(u) \leq d_\markl$, and $\dip(v) \leq d_\markr$.
	Furthermore, we can also decide in $\PSPACE$ whether a derivation with the above properties exists
	that also satisfies $\offset(u) > 0$.
\end{restatable}
\begin{proof}[Proof sketch]
We first construct the $\CFG$ $\cG_A$
for the language of pumps of the nonterminal $A$, meaning for
$L(\cG_A) = \{u\$v \mid A \derivs_\cG uAv \}$.
Then we intersect with the regular language $\Theta^*a\Theta^*\$\Theta^*$,
and apply \cref{lem:PAformulaEffect} to the resulting grammar.
This is possible, because tame-pumping implies that the grammar for the pumps
has a uniform offset of zero.
We can modify the resulting Presburger formula from \cref{lem:PAformulaEffect} to check for the required dips,
and modify it further to check for the positive offset for $u$.
Finally, we use the fact that testing satisfiability of an existential Presburger formula
is in $\NP$ \cite{BoroshTreybig76}.
\end{proof}

Now we are almost ready to prove \cref{lem:convertTosECFG}.
The last thing we need is for an $\sECFG$ to perform $\PSPACE$-computations on paths
in its derivation trees:

\begin{observation}\label{ECFG-PSPACE}
  An $\sECFG$ can simulate $\PSPACE$-computations on exponentially long paths in its
  derivation trees.
  This is because the nonterminals are polynomially long strings and can therefore act as
  polynomial space Turing tape configurations.
  Moreover, the transducers of the $\sECFG$ can easily be constructed to enforce the step-relation of a Turing machine.
  If we apply this enforcement to productions of the form $A \rightarrow B$, then the path that
  simulates the $\PSPACE$-computation will not even have any additional side paths until
  after the computation is complete.
  Thus, only the result of the computation will affect the derived word.
  
  Since grammars and transducers are non-deterministic
  (and $\NPSPACE = \PSPACE$), 
  we can even implement non-determinism and guessing within such computations.
\end{observation}

We are ready to present a proof sketch of \cref{lem:convertTosECFG}.
The main idea is that $\cG'$ simulates derivation trees of $\cG$ by 
keeping track of at most polynomially many nodes, and abstracting away pumps via the
previous auxiliary results.

If a nonterminal $A$ of $\cG$ does not belong to $N_0$ (i.e., it produces a marker),
then $\cG'$ guesses a production $A \rightarrow BC$ to apply.
If $A$ furthermore belongs to $N_{\#\bar \#}$,
then $\cG'$ also guesses a pump to apply in the form of a $4$-tuple consisting of two
dip values $d_\markl,d_\markr \in \N$ and two alphabets $\Gamma_\markl,\Gamma_\markr \subseteq \Gamma$.
Guessing and storing the dip values is possible in $\PSPACE$,
since they are exponentially bounded by \cref{lem:checkTamePumping}~(2).
For each $a \in \Gamma_\markl$, \cref{dips-for-spawns} is used on input
$A,a,d_\markl,d_\markr$ to check in $\PSPACE$ whether a matching pump exists.
A symmetric version of \cref{dips-for-spawns} is also used for each $a \in \Gamma_\markr$.
Then, if all checks succeed, $\cG'$ simulates the pump as
$A \rightarrow \bar{x}^{d_\markl}x^{d_\markl}\Gamma_\markl^*
BC \bar{x}^{d_\markr}x^{d_\markr}\Gamma_\markr^*$.
This simulation clearly preserves minimal dips and handler names,
whereas by tame-pumping the combined offset of a pump is zero anyway,
and therefore need not be computed.

If a nonterminal $A$ belongs to $N_0$, then $\cG'$ abstracts away
its entire subtree.
To this end it generates a pumpfree subtree on-the-fly using depth-first search,
which is possible in $\PSPACE$ since without pumps the tree has polynomial height.
During this process pumps are simulated using the same strategy as before.

We also need to ensure that nonterminals of $\cG'$ in $N_0$ only have productions that
allow for a single leaf node below them. 
To this end $\cG'$ only ever derives letters and alphabets $\Gamma'^*$ one at a time.
Consider the up to two \emph{main paths} in a derivation tree of $\cG'$,
by which we mean the paths leading from the root to a marker.
Whenever $\cG'$ simulates a pump as $A \rightarrow u'Av'$ in the above process,
it extends the main path by $|uv|$ and in each step only derives a single nonterminal from $N_0$
to the left or right.
When $\cG'$ abstracts an entire subtree of a nonterminal in $N_0$,
then this subtree is also produced to the left or right of the main path, without leaving said path.

Additionally, whenever $\cG'$ simulates a pump of some $A$, then $\cG'$ assumes that this pump is the combination of all pumps that occur in the original derivation tree for that instance of $A$.
Thus, below such a pump, it remembers in polynomial space, that $A$ is not allowed to occur anymore.
Finally, whenever $\cG'$ checks by \cref{dips-for-spawns} that a pump exists with
$\offset(u) > 0$, then this is a so-called increasing pump, and it can be repeated to achieve
an infix with arbitrary high offset.
Thus, dip values below this pump cannot make up for this offset and therefore will no
longer be simulated.

\subsection{Abstracting divided pumps}
\label{sec:absSplitPumps}
We have now removed all the undivided pumps and are left with derivation trees as in the middle picture of \cref{fig:abstracting-pumps}.
In this subsection, we will show the following:
\begin{restatable}{lemma}{sECFGtoDCAut}
	\label{lem:sECFGtoDCAut}
	Given a tame-pumping almost-pumpfree $\sECFG$ $\cG$ with $L(\cG) \subseteq \Decomps$, one can construct in polynomial space a $\dsNFA$ $\calB$ such that $\dc{L(\calB)}=\dc{L(\cG)}$
and $|\calB|$ is polynomially bounded in $|\cG|$.
\end{restatable}
We give a proof sketch here, the details can be found in the full version of the paper.
Our starting point in the proof of \cref{lem:sECFGtoDCAut} is the following key observation: The offsets which occur during the production of any \emph{admissible} marked word $w$ which contains exactly one marker are bounded. This allows us to keep track of the offset precisely, which is necessary for us to solve the marked Dyck factor (MDF) problem. 

For a node $t$ in a derivation tree $T$, let $w(t)$ denote the word derived by the subtree rooted at $t$ and let $u(t)=\inside{w(t)}$, $v(t)=\outside{w(t)}$. 
\begin{restatable}{lemma}{boundingOffsetInUPFtree}
	\label{lem:boundingOffsetInUPFtree}
	There exists a polynomial $p$ such that for any uniformly marked,
	tame-pumping, almost-pumpfree $\sECFG$ $\cG$ the following holds.  Let $T$ be a
	derivation tree of $\cG$ which produces an admissible marked word containing
	$\#$ or $\bar{\#}$, but not both. Then we have $|\offset(u(t))|,
	|\offset(v(t))|\le 2^{p(|\cG|)}$. 
\end{restatable} 

\begin{proof}
	\newcommand{\UsumPumps}{\sum_{i=1}^k \offset(\hat{u}_i)}
	\newcommand{\VsumPumps}{\sum_{i=1}^k \offset(\hat{v}_i)}
	We consider the case when the word derived is of the form $u\#v$, the case for $v \bar{\#} w$ being symmetric.
	Our derivation tree $T$ has a skeleton $T'$ into which pumps are inserted to form $T$. This means $u\#v=u'_k\hat{u}_k\cdots u'_1\hat{u}_1u'_0\#v'_0\hat{v}_1v'_1\cdots \hat{v}_kv'_k$, where $u'_k\cdots u'_0\#v'_0\cdots v'_k$ is the word generated by $T'$ and each pair $(\hat{u}_i,\hat{v}_i)$ is derived using a pump. Then we have
	\begin{align*}
		\offset(u)&=\overbrace{\offset(u'_k\cdots u'_0)\vphantom{\UsumPumps}}^{=:U_0} + \overbrace{\UsumPumps}^{=:U_1}, \\
		\offset(v)&=\underbrace{\offset(v'_0\cdots v'_k)\vphantom{\VsumPumps}}_{=:V_0} + \underbrace{\VsumPumps}_{=:V_1}.
	\end{align*}
	We claim that each of the numbers $|U_0|,|U_1|,|V_0|,|V_1|$ is
	bounded by $n(\cG)$, the number of nonterminals of $\cG$.
  This clearly implies the
	\lcnamecref{lem:boundingOffsetInUPFtree}: Since $\cG$ is a succinct grammar, it
	has at most exponentially many nonterminals in the size of its description.
	We begin with $U_0,V_0$. The tree $T'$ contains each nonterminal of $\cG$ at most
	once, and by property~(C2) in \cref{it:specialPathC2}, we know that the subtree
	under each nonterminal in $T'$ not containing $\#$ has offset $-1$, $0$, or
	$1$. Thus, $|U_0|,|V_0|\le n(\cG)$. 
	The bound on $|U_1|,|V_1|$ is due to admissibility of $u\#v$: It yields $V_0+V_1=\offset(v)\ge
	0$ and thus $V_1\ge -V_0$. Moreover, by tame-pumping, we know that
	$\offset(\hat{v}_i)\le 0$ for each $i\in[1,k]$, and thus $V_1\le 0$. Together, we
	obtain $V_1\in [-V_0,0]$. Finally, tame-pumping also implies
	$\offset(\hat{u}_i)=-\offset(\hat{v}_i)$ for each $i\in[1,k]$ and hence $U_1=-V_1$.
\end{proof}

\begin{remark}
	Note that the bound only holds under the condition of admissibility. An easy counterexample is the tame-pumping language $L=\{ x^n\#\bar{x}^n \mid n \in \mathbb{N} \}$.
\end{remark}

The $\dsNFA$ $\calB$ of \cref{lem:sECFGtoDCAut} can now be constructed in three steps as follows:

\subparagraph{Step I: Tracking counter effects.} We first observe that since $\cG$ is almost-pumpfree, its pumps $A\derivs uAv$ can be simulated by a transducer that traverses the derivation tree bottom-up. Thus, we can construct a \emph{singly} succinct finite-state transducer $\calT_A$ with size polynomial in $|\cG|$ that captures all pumps $A\derivs uAv$. To be precise, $\calT_A$ accepts exactly those pairs $(u,v)$ for which $A\derivs u^{\rev}Av$. The transducer $\calT_A$ has one state for each nonterminal of $\cG$. 

Since $\calB$ will need to preserve offset and dip, we need to expand $\calT_A$ to track them as well. Here, it is crucial that we only need to do this for $A\in N_{\#}\cup N_{\bar{\#}}$ and pumps $A\derivs uAv$ that are used to derive an admissible word. According to \cref{lem:boundingOffsetInUPFtree} tells us that in such a pump, the absolute values of offsets and dips of $u$ and $v$ are bounded by $2^{q(|\cG|)}$ for some polynomial $q$. Thus, we can modify $\calT_A$ so as to track the dip and offset of the two words it reads. Therefore, for each $A\in N_{\#}\cup N_{\bar{\#}}$ and each quadruple $\bx=(d_{\markl},\delta_{\markl},d_{\markr},\delta_{\markr})$ of numbers with absolute value at most $2^{q(\cG)}$, we can construct in $\PSPACE$ a transducer $\calT_{A,\bx}$ with
\begin{align*}
 \text{$(u,v)$ is accepted by $\calT_{A,\bx}$~~~iff~~~} &\text{$A\derivs u^{\rev}Av$ and $e(u^{\rev})=(d_{\markl},\delta_{\markl})$, and $e(v)=(d_{\markr},\delta_{\markr})$}. 
\end{align*}
Moreover, $\calT_{A,\bx}$ is singly succinct, polynomial-size, and can be computed in $\PSPACE$. Observe that by \cref{lem:boundingOffsetInUPFtree}, if a pump $A\derivs uAv$ is used in a derivation of an admissible word, then for some quadruple $\bx$, the pair $(u^{\rev},v)$ is accepted by $\calT_{A,\bx}$.

\subparagraph{Step II: Skeleton runs.} The automaton $\calB$ has to read words from left to right, rather than two factors in parallel as $\calT_{A}$ and $\calT_{A,\bx}$ do. To this end, it will guess a run of $\calT_{A,\bx}$ without state repetitions; such a run is called a \emph{skeleton run}. For a fixed skeleton run $\rho$, the set of words read in each component of $\calT_{A,\bx}$ is of the shape $\Gamma_0^*\{a_1,\varepsilon\}\Gamma_1^*\cdots \{a_k,\varepsilon\}\Gamma_k^*$, where each $a_i$ is read in a single step of $\rho$ and $\Gamma_i$ is the set of letters from $\Gamma$ seen in cycles in a state visited in $\rho$. Sets of this shape are called \emph{ideals}~\cite{GHKNS-til2020}. The ideal for the left (right) component is called the \emph{left} (\emph{right}) \emph{ideal} of the skeleton run. Note that since $\calT_{A,\bx}$ has exponentially many states, the skeleton run is at most exponentially long.

\subparagraph{Step III: Putting it together.} The $\dsNFA$ $\calB$ guesses and
verifies an exponential size skeleton $T$ of the $\sECFG$ $\cG$. Moreover, for
each node $t$ that is above $\#$ or $\bar{\#}$---but not both---it guesses a
quadruple $\bx=(d_{\markl},\delta_{\markl},d_{\markr},\delta_{\markr})$ with
$d_{\markl},d_{\markr}\in[0,2^{q(|\cG|)}]$,
$\delta_{\markl},\delta_{\markr}\in[-2^{q(\cG)},2^{q(\cG)}]$ and a skeleton run $\rho_t$ of the transducer $\calT_{A,\bx}$, where $A$ is $t$'s
label.
The automaton $\calB$ then traverses the
skeleton $T$ in-order; i.e.\ node, left subtree, right subtree, node; meaning
each inner node is visited exactly twice.
Whenever $\calB$
visits a node $t$ as above, it produces an arbitrary word from an ideal
 of $\rho_t$: For the first (resp.\ second) visit of $t$, it uses the left (resp.\ right) ideal of
$\rho_t$. Moreover, in addition to the word from the left ideal, $\calB$ outputs a string $w\in\{x,\bar{x}\}^*$ with $e(w)=(d_\markl,\delta_\markl)$, where $\bx=(d_{\markl},\delta_{\markl},d_{\markr},\delta_{\markr})$ is the quadruple guessed for $t$ (and similarly for the right ideal).  This way, it preserves offset and dip at the separators $\#$ and $\bar{\#}$.

Since the skeleton $T$ has exponentially many nodes (in $|\cG|$) and each skeleton run $\rho_t$ requires exponentially many bits, the total number of bits that $\calB$ has to keep in memory is also bounded by an exponential in $|\cG|$.

\label{beforebibliography}
\newoutputstream{pages}
\openoutputfile{main.pages.ctr}{pages}
\addtostream{pages}{\getpagerefnumber{beforebibliography}}
\closeoutputstream{pages}
\bibliography{bibliography}

\appendix

\section{Results from Section~\ref{sec:introduction}}

In \cref{sec:introduction}, we make the following observation when speaking about $\VASS$ language inclusion:
checking whether a given $\VASS$ language is included in the set of \emph{prefixes} of the one-letter Dyck language is already equivalent to $\VASS$ reachability.
In the following we show Ackermann completeness of the aforementioned inclusion problem.
The proof employs reductions to and from reachability for $\VASS$, therefore also proving
equivalence between these two problems.

\begin{proposition}
It is Ackermann-complete to decide whether a given $\VASS$ $\cV$
satisfies $L(\cV) \subseteq \mathsf{Pref}(\Dyck_{\{x\}})$,
where $\mathsf{Pref}(\Dyck_{\{x\}})$ is the set of prefixes of words in $\Dyck_{\{x\}}$.
\end{proposition}

\begin{proof}
For the lower bound we reduce from the reachability problem for $\VASS$,
which is Ackermann-complete~\cite{CzerwinskiPNAckermann2021,lerouxReachabilityProblemPetri2022}.
A similar statement and proof idea can be found in \cite[Lemma~12]{CzerwinskiH22}.
The reachability problem asks whether a given $\VASS$ $\cV = (Q,\{x,\bar x\},I,E,q_0,q_f)$ has a run $(q_0,\vecz) \xrightarrow{w} (q_f,\vecz)$
(the input letters are irrelevant here).
We transform $\cV$ into a new $\VASS$ $\cV'$ which makes the counter changes visible in the input word over $\{x,\bar x\}$:
Each increment on one of the counters is translated to reading input $x$, and decrements are translated to reading $\bar x$.
Observe that the coverability language satisfies $L(\cV') \subseteq \mathsf{Pref}(\Dyck_{\{x\}})$,
and that $L(\cV')$ intersects $\Dyck_{\{x\}}$ if and only if 
$\cV$ has a run $(q_0,\vecz) \xrightarrow{w} (q_f,\vecz)$.
Therefore, $\cV$ has \emph{no} run $(q_0,\vecz) \xrightarrow{w} (q_f,\vecz)$
if and only if $L(\cV') \bar x \subseteq \mathsf{Pref}(\Dyck_{\{x\}})$.
It is easy to construct a $\VASS$ $\cV''$ with $L(\cV'') = L(\cV') \bar x$, which completes the reduction.

The Ackermann upper bound can be shown using a very simple reduction to the reachability problem, which can be decided in Ackermann complexity~\cite{LerouxSchmitz2019}. An Ackermann upper bound for a much more general inclusion problem is shown in~\cite[Theorem 2]{CzerwinskiH22}.
\end{proof}

\section{Results from Section~\ref{sec:asyncp}}
\label{app:useful}

Let $\asyncp=(Q, \Sigma, \Gamma, \cG, \Delta, q_0, q_f, a_0)$ be a program and
$\cG = (N,\Sigma \cup \Gamma,P,S)$ be its grammar.
Recall that a nonterminal $B \in N$ is called \emph{useful}
if there exists a run $\rho$ of $\asyncp$ such that $\rho$ reaches $q_f$
and there exists a derivation in $\rho$, whose derivation tree contains $B$.
In \cref{sec:asyncp} we claimed that one can compute the set of useful nonterminals
of a given program $\asyncp$ in exponential space.
We show this in the following.

\begin{lemma}
\label{lem:computeProdNonterms}
	Given a program $\asyncp=(Q, \Sigma, \Gamma, \cG, \Delta, q_0, q_f, a_0)$, we can compute the set of its useful nonterminals in $\EXPSPACE$.
\end{lemma}
\begin{proof}
It suffices to show that we can check whether a particular nonterminal $X$ is useful. In order to do so, we reduce the problem to global state reachability in an associated program $\asyncp_X$. Program $\asyncp_X$ has two copies of $Q$. It has a copy of $\Delta$ per copy of $Q$ with the first copy of $q_0$ being its start state and the second copy of $q_f$ being its final state. It begins by simulating $\asyncp$ in the first copy. Whenever a production that uses $X$ is applied, $\asyncp_X$ additionally spawns a task $s$ with $s \not\in \Gamma$ such that when $s$ is run, it allows $\asyncp_X$ to move from the first to the second copy of $Q$. Program $\asyncp_X$ then continues to simulate $\asyncp$ in the second copy till it reaches its final state.

Let $\asyncp = (Q, \Sigma, \Gamma, \cG, \Delta, q_0, q_f, a_0)$ be a program,
$\cG = (N,\Sigma \cup \Gamma,P,S)$ its grammar,
and $X \in N$ a nonterminal.
To check if $X$ is useful, we formally construct the new program
$\asyncp_X=(Q', \Sigma, \Gamma', \cG', \Delta', q'_0, q'_f, a_0)$ as follows:
\begin{itemize}
	\item $Q'=Q \times \{ 0,1\}$,
	\item $\Gamma'=\Gamma \cup \{ s\}$ where $s \not\in \Gamma$,
	\item $q'_0=(q_0,0), q'_f=(q_f,1)$,
  \item $G' = (N \cup \{A_s\},\Sigma \cup \Gamma',P',S)$ where $P'$ is constructed from $P$
  by adding the production rule $A_s \rightarrow \varepsilon$ and changing
  every rule with $X$ on the left hand side from $X \rightarrow \alpha$ to $X \rightarrow \alpha s$, and
	\item $\Delta'$ contains the following: within each of the two copies of $Q$, the rules are inherited from $\Delta$, and for each $q \in Q$ we add the rule $(q,0) \xhookrightarrow{s,A_s} (q,1)$ (where $L(G',A_s) = \{\varepsilon\}$ by construction of $G'$).
\end{itemize}
The constructed program $\asyncp_X$ is such that $q'_f$ is reachable iff $X$ is useful in $\asyncp$. By \cref{thm:safetyAsync}, the lemma follows.
\end{proof}

\section{Results from Section~\ref{sec:algorithm}}
\label{sec:app-algo}

\subparagraph*{Tame-pumping grammars and programs}
Here we prove the following:

\checktamepumping*

\begin{proof}
\emph{Part 1.}
Given a context-free grammar $\calG=(N,\{x,\bar x\},P,S)$ in Chomsky normal form
Let $x_{\markl}, \bar x_{\markl}$ and $x_{\markr}, \bar x_{\markr}$ be copies of $x, \bar x$,
and for a word $u = a_1 \cdots a_n \in \{x,\bar x\}^*$
we define $u_{\markl} = (a_1)_{\markl} \cdots (a_n)_{\markl}$ and $u_{\markr} = (a_1)_\markr \cdots (a_n)_\markr$.
We will construct a grammar $\calG_{A_0}$ for the language $L(\calG_{A_0}) = \{ u_{\markl} v_{\markr} \mid A_0 \derivs u A_0 v \}$.
Recall that \emph{Presburger arithmetic} is the first-order theory of $(\N,+,<,0,1)$.
By~\cite{vermaComplexityEquationalHorn2005}, we can compute an existential Presburger formula
$\phi_X$ in polynomial time which defines the Parikh image of $L(\calG_{A_0})$.
From there we can express the desired properties in existential Presburger arithmetic.
Since the truth of existential Presburger sentences is known to be in $\NP$ \cite{BoroshTreybig76}, the lemma follows. 

The grammar $\calG_{A_0}=(N',\{x_{\markl},\bar x_{\markl},x_{\markr},\bar x_{\markr}\},P',A_0)$ is constructed as follows:
Its set of nonterminals $N'=\{ A, A_{\markl},
 A_{\markr} \mid A \in N\}$ 
 contains three copies of each nonterminal in $N$.
It contains the following productions:
\begin{itemize}
	\item for each rule $A \rightarrow BC$ in $P$, add $A_{\markl} \rightarrow B_{\markl} C_{\markl}$,
	$A_{\markr} \rightarrow B_{\markr} C_{\markr}$, $A \rightarrow B_{\markl} C$ and $A \rightarrow B C_{\markr}$,
	\item for each rule $A \rightarrow a$ where $a \in \{x,\bar x\}$
	add $A_{\markl} \rightarrow x_{\markl}$ and $A_{\markr} \rightarrow x_{\markr}$, and 
	\item add $A_0 \rightarrow \varepsilon$,
\end{itemize}
There exists a violation for the tame-pumping property if and only if
there exists a word $w \in L(\calG_{A_0})$ with $|w|_{x_{\markl}} - |w|_{\bar x_{\markl}} \neq |w|_{x_{\markr}} - |w|_{\bar x_{\markr}}$
or $|w|_{x_{\markl}} < |w|_{\bar x_{\markl}}$.
There exists a zero pump (increasing pump) if and only if there is a word $w \in L(\calG_{A_0})$
with $|w|_{x_{\markl}} - |w|_{\bar x_{\markl}} = 0$ ($|w|_{x_{\markl}} - |w|_{\bar x_{\markl}} > 0$).
We already mentioned above that we can construct in polynomial time an existential
Presburger formula for the Parikh image of $L(\calG_{A_0})$, i.e.\ the set
$\{w_{x_{\markl}},w_{\bar x_{\markl}},w_{x_{\markr}},w_{\bar x_{\markr}} \mid w \in L(\calG_{A_0}) \}$.
With this, these statements can be easily expressed in existential Presburger arithmetic.

\emph{Part 2.} We first bring $\calG$ into Chomsky normal form, which increases the size only polynomially.
First, we claim that $|\offset(w)| \le 2^{|\calG|}$ for all $w \in L(\calG,A)$ and all nonterminals $A$.
If $T$ is a derivation tree in $\calG$ then removing a pump from $T$ does not change the offset
of the yield word since $\calG$ is tame-pumping.
After repeatedly removing pumps until no more exist, we obtain a derivation tree whose height is bounded by the number of nonterminals.
Therefore, its size and the absolute value of its offset are bounded by $2^{|\calG|}$.

Recall that $\dip(w)$ is defined as $-\min\{\offset(w_1) \mid w_1$ is a prefix of $w \}$.
Now, let $T$ be a derivation tree for a word $w \in L(\calG)$
and let $w_1$ be some prefix of $w$.
Consider the path $P$ from the root to the leaf $\ell$, producing the last letter of $w_1$.
We remove pumps on the path $P$, which does not increase the offset
of the prefix $w_1$ produced between the leftmost leaf and the leaf $\ell$ because $\calG$ is tame-pumping.
After repeatedly removing pumps we can ensure that the length of $P$ is bounded by the number of nonterminals in $\calG$.
Suppose that $A_1, \dots, A_n$ are the nonterminals branching off to the left on $P$
together with the parent node of $\ell$.
We can write $w_1$ as $u_1 \cdots u_n$ where $u_i \in L(\calG,A_i)$ for all $1 \le i \le n$ and $n \le |\calG|$.
Therefore $|\offset(w_1)| \le \sum_{i=1}^n |\offset(u_i)| \le |\calG| \cdot \max_{i=1}^n |\offset(u_i)| \le |\calG| \cdot 2^{|\calG|} \le 2^{2 \cdot |\calG|}$.
With the polynomial $p$ defined as $p(y) := 2y$ the statement then follows.

\emph{Part 3.}
Suppose that there exists a pump $A \derivs uAv$ where $\offset(u) \neq \offset(v)$ or $\offset(u) < 0$.
Consider an accepting run $\rho$ of $\asyncp$ which uses the nonterminal $A$ in a particular derivation tree $T$.
By inserting the pump repeatedly into $T$ we can produce an offset or dip violation:
First, observe that inserting pumps, can only enlarge the multiset of handler names produced in $T$.
If $\offset(u) \neq \offset(v)$ then a single insertion of the pump changes the offset of the word read in $T$,
which implies $L(\asyncp) \not \subseteq D_X$.
If $\offset(u) < 0$ then, by inserting the pump sufficiently often into $T$,
we obtain a word with negative dip, which also implies $L(\asyncp) \not \subseteq D_X$.
\end{proof}

\subparagraph*{From programs to doubly succinct $\VASS$}
We would like to prove \cref{cor:ap-to-vass}, which states that for a
given asynchronous program we can construct in polynomial space a $\dsVASS$ such
that their respective languages are downward closure equivalent regarding the ordering
$\extsw$ on the set of admissible words $\AdmDecomps$.
To this end, we first need the following auxiliary result:

\begin{lemma}
	\label{lem:extsw-compatible}
	Let $u_1,u_2,v_1,v_2 \in \AdmDecomps$ with $u_1u_2 \in \AdmDecomps$.
	If $u_1 \extsw v_1$ and $u_2 \extsw v_2$ then $u_1u_2 \extsw v_1v_2$.
\end{lemma}

\begin{proof}
First, observe that $u_1 \extsw' v_1$ and $u_2 \extsw' v_2$ implies $u_1 u_2 \extsw' v_1 v_2$
because the same holds for the subword order $\subword$ and the syntactic order $\sdyck$.
Now assume that $u_1 \extsw v_1$ and $u_2 \extsw v_2$ and $u_1u_2 \in \AdmDecomps$.
Then $v_1v_2$ is also admissible and contains the same markers as $u_1u_2$.
For an exhaustive proof, we would need to do a (simple but tedious) case distinction,
depending on the markers in $u_1, u_2, v_1, v_2$,
to prove that $u_1u_2 \extsw v_1v_2$.
For example, consider the case where
$u_1$ (and therefore also $v_1$) only contains the marker $\#$, and $u_2$ (and therefore also $v_2$)
only contains the marker $\bar \#$.
Then $\inside{u_1 u_2} = \inside{u_1} \, \inside{u_2} \extsw' \inside{v_1} \, \inside{v_2} = \inside{v_1 v_2}$
and $\outside{u_1 u_2} = \outside{u_1} \, \outside{u_2} \extsw' \outside{v_1} \, \outside{v_2} = \outside{v_1 v_2}$.
The other cases are similar.
\end{proof}

Now we are ready to prove the following.

\aptovass*

\begin{proof}
  Let $\asyncp = (Q, X \cup \bar X, \Gamma, \cG, \Delta, q_0, q_f, \gamma_0)$ be the asynchronous program.
  For each nonterminal $A$ we apply \cref{thm:extended-downclosure} to the grammar $\calG$
  with starting nonterminal $A$
  and obtain a $\dsNFA$ $\calB_A$ with $\dc{L(\calG,A)} = \dc{L(\calB_A)}$.
  
  The set of counters of the $\VASS$ $\calV$ is the set of handler names $\Gamma$.
  The state set of $\calV$ is the disjoint union of an initial state $q_{\mathsf{in}}$,
  the state set $Q$ of $\asyncp$, and $Q \times P$
  where $P$ is the disjoint union of all state sets of the $\dsNFA$s $\calB_A$.
  The final state is $q_f$.
  For $a \in \Gamma$ we use $\be_a$ to denote the vector with a $1$ at the coordinate
  corresponding to $a$ and a $0$ everywhere else.
  The transitions of $\calV$ are given as follows:
  \begin{itemize}
    \item There is a transition $q_{\mathsf{in}} \xrightarrow{\be_{\gamma_0}} q_0$.
    \item Each transition $q \xhookrightarrow{a,A} q'$ in $\asyncp$
    is translated into transitions $q \xrightarrow{-\be_a} (q',p_0)$ and $(q',p_f) \xrightarrow{\vecz} q'$
    where $p_0$ (respectively $p_f$) is the unique initial (respectively final) state of $\calB_A$.
    \item Each transition $p \xrightarrow{a} p'$ where $a \in \Gamma$ in a $\dsNFA$ $\calB_A$ is translated
    into transitions $(q,p) \xrightarrow{\varepsilon, \be_a} (q,p')$ for all $q \in Q$.
    \item Each transition $p \xrightarrow{x} p'$ where $x \in X \cup \bar X$ in a $\dsNFA$ $\calB_A$ is translated
    into transitions $(q,p) \xrightarrow{x, \vecz} (q,p')$ for all $q \in Q$.
  \end{itemize}
  We claim that $\dc{L_\asyncp((q,\mmap))} = \dc{L_\calV((q,\mmap))}$
  for all configurations $(q,\mmap) \in Q \times \multiset{\Gamma}$,
  where $L_\asyncp(c)$ and $L_\calV(c)$ denote the language recognized by $\asyncp$ and $\calV$
  from start configuration $c$.
  We will proceed inductively on the word length.
  
  For the $\subseteq$-direction,
  consider an accepting run $(q,\mmap) \xrightarrow{u} (q',\mmap') \xrightarrow{v_1} \cdots \xrightarrow{v_\ell} (q_f,\mmap_f)$
  of $\asyncp$ on an admissible word $uv$ where $v = v_1 \cdots v_\ell$.
  Hence, there exists a transition $q \xhookrightarrow{a,A} q'$ and a word $w \in L(\calG, A)$ such that 
  $u = \pi_\Sigma(w)$ and $\mmap' = (\mmap \ominus \multi{a}) \oplus \Parikh(\pi_\Gamma(w))$.
  Since $L(\calG,A)$ and $L(\calB_A)$ have the same downward closure,
  there exists $\hat w \in L(\calB_A)$ such that $w \extsw \hat w$.
  Therefore $\calV$ contains a run $(q,\mmap) \xrightarrow{\hat u} (q',\mmap'')$
  where $\hat u = \pi_{X \cup \bar X}(\hat w)$ and $\mmap'' \ge \mmap'$.
  Moreover, by induction hypothesis there is a word $\hat v \in L_\calV((q',\mmap'))$
  such that $v \extsw \hat v$.
  Since $L_\calV((q',\mmap')) \subseteq L_\calV((q',\mmap''))$
  we obtain an accepting run of $\calV$ on $\hat u \hat v$.
  Since $u \extsw \hat u$, \cref{lem:extsw-compatible} implies that $u v \extsw \hat u \hat v$.
  
  The $\supseteq$-direction is analogous:
  Consider an accepting run of $\calV$ starting from $(q,\mmap)$
  and decompose it into minimal subruns which start and end in configurations from $Q \times \multiset{\Gamma}$,
  say $(q,\mmap) \xrightarrow{u} (q',\mmap') \xrightarrow{v_1} \cdots \xrightarrow{v_\ell} (q_f,\mmap_f)$.
  By definition of $\calV$ there exists a transition $q \xhookrightarrow{a,A} q'$ and a word $w \in L(\calB_A)$ such that 
  $u = \pi_\Sigma(w)$ and $\mmap' = (\mmap \ominus \multi{a}) \oplus \Parikh(\pi_\Gamma(w))$.
  Since $L(\calG,A)$ and $L(\calB_A)$ have the same downward closure,
  there exists $\hat w \in L(\calG,A)$ such that $w \extsw \hat w$.
  Therefore $\asyncp$ contains a step $(q,\mmap) \xrightarrow{\hat u} (q',\mmap'')$
  where $\hat u = \pi_{X \cup \bar X}(\hat w)$ and $\mmap'' \ge \mmap'$.
  Moreover, by induction hypothesis there is a word $\hat v \in L_\calV((q',\mmap'))$
  such that $v \extsw \hat v$.
  Since $L_\asyncp((q',\mmap')) \subseteq L_\asyncp((q',\mmap''))$
  we obtain an accepting run of $\asyncp$ on $\hat u \hat v$.
  Since $u \extsw \hat u$, \cref{lem:extsw-compatible} implies that $u v \extsw \hat u \hat v$.
\end{proof}

\subparagraph*{Constructing a program for each type of violation}
As part of the algorithm presented in \cref{subsec:algorithm} we construct
the programs $\asyncp_\ioff$, $\asyncp_\idip$, $\asyncp_\imis$ recognizing languages
over the alphabet $\{ x, \bar x, \#, \bar \# \}$ given by the following equation
(\cref{eq:aux-programs}):
\begin{equation*}
  \begin{aligned}
    L(\asyncp_\ioff) &= \{ \rho(w) \mid w \in L(\asyncp) \}, \\
    L(\asyncp_\idip) &= \{ \# \rho(v) \bar \# \rho(\bar y w) \mid v \bar y w \in L(\asyncp) \text{ for some } v,w \in (X \cup \bar X)^*,\,y \in X \}, \\
    L(\asyncp_\imis) &= \{ \rho(u) \# \rho(v) \bar \# \rho(w) \mid u y v \bar z w \in L(\asyncp), \\ & \qquad\qquad\qquad\qquad\qquad\qquad \text{for some } u,v,w \in (X \cup \bar X)^*, \, y \neq z \in X \}.
  \end{aligned}
\end{equation*}
Here $\rho \colon (X \cup \bar X)^* \to \{x,\bar x\}^*$ is the morphism 
which replaces all letters in $X$ by unique letter $x$ and all letters in $\bar X$ by unique letter $\bar x$.
In \cref{subsec:algorithm} we also mention that not only are these constructions possible
in polynomial time, but furthermore in all three cases tame-pumping is preserved:
if the original asynchronous program $\asyncp$ is tame-pumping,
then we can construct programs $\asyncp_\ioff$, $\asyncp_\idip$, $\asyncp_\imis$
that are also tame-pumping.
We prove this in the following.

\begin{restatable}{lemma}{computeaux}
\label{cor:compute-aux}
Given a tame-pumping asynchronous program $\asyncp$, we can construct in polynomial time
tame-pumping asynchronous programs $\asyncp_\ioff$, $\asyncp_\idip$, $\asyncp_\imis$ for the languages in \cref{eq:aux-programs}.
\end{restatable}

\begin{proof}
A binary transduction $T$ is \emph{offset-preserving} if $\offset(u) = \offset(v)$ for all $(u,v) \in T$.
It is easy to see that each of the three output languages can be obtained by applying an offset-preserving
rational transduction to $L(\asyncp)$, namely $T_\ioff = \{ (u,\rho(u)) \mid u \in \{x, \bar x\}^* \}$,
\[
	T_\idip = \{ (v \bar y w,\# \rho(v) \bar \# \rho(\bar y w)) \mid v,w \in (X \cup \bar X)^*, \, y \in X \},
\]
and
\[
	T_\imis = \{ (u y v \bar z w, \rho(u) \# \rho(v) \bar \# \rho(w)) \mid u,v,w \in (X \cup \bar X)^*, \, y \neq z \in X \}.
\]
It is easy to show that given a rational transduction $T$
and an asynchronous program $\asyncp$, one can compute an asynchronous program $\asyncp_T$
such that $L(\asyncp_T) = T(L(\asyncp))$, using a standard ``triple construction'' on the grammar of $\asyncp$.
We now show that if $T$ is offset-preserving, and $\asyncp$ is tame-pumping,
then $\asyncp_T$ also is tame-pumping.

Suppose that $\calT$ is a finite-state transducer for $T$ with state set $P$ and initial state $p_0$ and final state $p_f$.
The state set of the new asynchronous program $\hat \asyncp$ is the product $Q \times P$.
The initial state is $(q_0,p_0)$ and the final state is $(q_f,p_f)$.
Each transition $q \xhookrightarrow{a,A} q'$ in $\asyncp$ is translated into transitions
$(q,p) \xhookrightarrow{a,A_{p,p'}} (q',p')$ where $q,q' \in Q$ and $p,p' \in P$.
Here a nonterminal $A_{p,p'}$ generates the language $\{ v \mid \exists u \colon p \xrightarrow{(u,v)}_\calT p', \, A \derivs u \}$.
Each production $A \to BC$ is translated into productions $A_{p,p'} \to B_{p,p''} C_{p'',p'}$ for $p,p',p'' \in P$,
and $A \to a$ is translated into productions $A_{p,p'} \to b$ if $p \xrightarrow{(a,b)}_\calT p'$.
Finally, we add productions $A_{p,p} \to \varepsilon$ for each nonterminal $A$ and each $p \in P$.

We claim that the resulting grammar is tame-pumping.
Consider a pump $A_{p,p'} \derivs u A_{p,p'} v$ in the new grammar. 
This means, that there exists a derivation $A \derivs wAz$ in $\calG$
and runs $p \xrightarrow{(w,u)} p$ and $p' \xrightarrow{(z,v)} p'$ in $\calT$.
Since $T$ is offset-preserving, we know $\offset(w) = \offset(u)$ and $\offset(z) = \offset(v)$.
Hence tame-pumping is transferred from $A \derivs wAz$ to $A_{p,p'} \derivs u A_{p,p'} v$.
\end{proof}

\section{Results from Section~\ref{ssec:absOrdPumps}}
\label{sec:app-undivided}

\subparagraph*{Succinct CFGs}
We use succinct representations of $\CFG$s and $\ECFG$s which represent nonterminals by polynomial-size strings. We assume rules allowed by Chomsky Normal Form, which are of the form $A \rightarrow BC$ or $A \rightarrow a$, but moreover also allow rules of the form $A \rightarrow B$. The succinct representation contains a ternary transducer for rules with nonterminals on the right hand side and a binary transducer for rules with terminals on the right hand side. The ternary transducer associates rules of the form $A \rightarrow B$ with triples in $N \times N \times 0^*$, where $N$ is the set of nonterminals. In the case of an $\ECFG$, there could also be extended productions of the form $A \rightarrow \Gamma_1^*$ which are also recognised by the binary transducer.

A succinct $\CFG$ ($\sCFG$) has nonterminals represented by polynomial size strings,
just as a succinct $\NFA$ has polynomial size strings representing states.
Formally, it is a tuple $\cH=(\calT_N,\calT_T,\Theta,\Lambda,w_0)$, where
$\calT_N$ is a length-preserving transducer with two tapes,
$\calT_T$ is a length-preserving transducer with three tapes,
$\Theta$ is a finite set of terminals with $0,1 \notin \Theta$,
$\Lambda$ is an alphabet encoding the nonterminals with $0 \in \Lambda$ and $\Theta \cap \Lambda = \emptyset$, and
$w_0 \in \Lambda^*$ is the start nonterminal.
Let $M = |w_0|$.
We assume $L(\calT_N) \subseteq \Lambda^M \times \Lambda^M \times \Lambda^M$
and $L(\calT_T) \subseteq \Lambda^M \times \Theta0^{M-1}$.

Like for succinct $\NFA$, there is a corresponding explicit context free grammar
$\cE(\cH)=(N,\Theta,P, S)$ with 
\begin{itemize}
  \item $N=\Lambda^M\setminus\set{0^M}$, 
  \item $w \rightarrow uv \in P$ for $u,v \in N$ iff $(w,u,v) \in L(\calT_N)$,
  \item $w \rightarrow u \in P$ for $u \in N$ iff $(w,u,0^M) \in L(\calT_N)$,
  \item $w \rightarrow a \in P$ for $a \in \Theta$ iff $(w,a0^{M-1}) \in L(\calT_T)$,
  \item $S=w_0$.
\end{itemize}
The language of $\cH$ is $L(\cH) = L(\cE(\cG))$, i.e. the language of its explicit $\CFG$.

The size of $\cH$ is defined as $|\cH| := |\calT_N| + |\calT_T| + |w_0|$.

We also consider succinct, extended $\CFG$ ($\sECFG$), where the explicit model is an $\ECFG$.
To encode the extended productions, we make a slightly different assumption on the language of $\cT_T$: 
$L(\calT_T) \subseteq \big(\Lambda^M \times \Theta 0^{M-1}\big) \cup \big(\Lambda^M \times (\Theta \cup \{0\})^{|\Theta|}10^{(M-1)-|\Theta|}\big)$.
Here a $1$ at position $|\Theta| + 1$ in the second component indicates that the pair 
corresponds to an extended production.
To ensure that we always have space for this $1$, we require $|w_0| > |\Theta|$ for $\sECFG$s.
Formally, in the explicit $\ECFG$, we have $w \rightarrow \Gamma^* \in P$ for $\Gamma \subseteq \Theta$
iff $(w,u10^{(M-1)-|\Theta|}) \in L(\calT_T)$ with $u \in (\Theta \cup \{0\})^{|\Theta|}$
such that $|u|_a \geq 1$ for each $a \in \Gamma$ and $|u|_b = 0$ for each $b \in \Theta \setminus \Gamma$.

\subparagraph*{Computing annotations}
In the proof sketch for \cref{lem:PAformulaEffect}, we mention that this directly follows
from \cite[Proposition~3.8]{BaumannGMTZ23} if we assume that the given offset-uniform
$\CFG$ $\cG$ is annotated. Here we call $\cG$ annotated if for every nonterminal $A$
the minimal dip that can be achieved by a word in $L(\cG,A)$ is given in the input,
denoted by $\mindip(A)$.
This means that if we can compute $\mindip(A)$ for every nonterminal $A$ of $\cG$ in
polynomial space then \cref{lem:PAformulaEffect} immediately follows also for the case
where the annotation is not given in the input.

\cref{lem:computingAnn} below states the desired result.
To prove it, we first need the following auxiliary result:

\begin{lemma}
\label{lem:mindip-rel}
Let $\cG = (N,\Sigma,P,S)$ be an offset uniform $\CFG$, and let
$A \to \alpha_1 \cdots \alpha_n \in P$ with $\alpha_i \in N \cup \Sigma$ for $i \in [1,n]$
be a production that occurs in some complete derivation of $\cG$.
Then the following holds:
\[
	\label{eq:dx}
	\mindip(A) \le \max_{1 \le k \le n} \left( \mindip(\alpha_k) - \sum_{i=1}^{k-1} \offset(\alpha_i) \right)
\] 
where $\mindip(x) = \dip(x)$ for $x \in \Sigma$, and $\offset(\cdot)$ denotes the offset of a terminal
symbol or the unique offset of a nonterminal.
\end{lemma}

\begin{proof}
Let us first argue that $\offset(\cdot)$ is well-defined on nonterminals $\alpha_i$.
Since $\cG$ is offset-uniform, so is the language of every nonterminal that occurs in some complete derivation.
As we assumed $A \to \alpha_1 \cdots \alpha_n$ to occur in some complete derivation, so do the $\alpha_i$.

Now, if a word $u \in \Sigma^*$ is factorized into $u = u_1 \cdots u_n$ then
\begin{equation}
	\label{eq:du}
	\dip(u) = \max_{1 \le k \le n} \left( \dip(u_k) - \sum_{i=1}^{k-1} \offset(u_i) \right).
\end{equation}
For $1 \le i \le n$, let $u_i$ be any word produced by $\alpha_i$ with $\dip(u_i) = \mindip(\alpha_i)$,
and define $u = u_1 \cdots u_n$.
Then we have:
\begin{align*}
	\mindip(A) &\le~ \dip(u)\\
  &= \max_{1 \le k \le n} \left( \dip(u_k) - \sum_{i=1}^{k-1} \offset(u_i) \right) = \max_{1 \le k \le n} \left( \mindip(\alpha_k) - \sum_{i=1}^{k-1} \offset(\alpha_i) \right).
\end{align*}
This concludes the proof.
\end{proof}

Now we are ready to prove that annotations are computable in polynomial space for offset-uniform $\CFG$s.

\begin{lemma}
\label{lem:computingAnn}
	Given an offset-uniform $\CFG$ $\cG = (N,\Sigma,P,S)$,
	one can compute the value $\mindip(A)$ for every nonterminal $A \in N$ in polynomial space.
\end{lemma}

\begin{proof}

In the following we can ignore all nonterminals that do not occur in some complete derivation.
As a first step, compute in polynomial time the value $\offset(A)$
for each nonterminal $A$ in a bottom-up fashion, see~\cite[Proof of Lemma~3.2]{BaumannGMTZ23}.
Every value $\mindip(A)$ is bounded by a number $M \in \N$,
which is exponentially large in the grammar size,
since $\mindip(A)$ is bounded by all dips among words derived by $A$, and
any (productive) nonterminal $A$ derives some word that is only exponentially long
(also follows from \cref{lem:checkTamePumping}~(2)).

The algorithm maintains a function $D \colon N \to \N$.
If $\alpha \in \Sigma$ we also write $D(\alpha)$ for $\dip(\alpha)$.
Furthermore, if $\alpha_1 \cdots \alpha_n \in (N \cup \Sigma)^*$ for $n \geq 2$, we define
\begin{equation}
	\label{eq:d-approx}
	D(\alpha_1 \cdots \alpha_n) = \max_{1 \le k \le n} \left( D(\alpha_k) - \sum_{i=1}^{k-1} \offset(\alpha_i) \right).
\end{equation}

\begin{enumerate}
\item Initialize $D(A) \leftarrow M$ for all $A \in N$.
\item While there exists a production $A \to w$ where $D(A) > D(w)$, set $D(A) \leftarrow D(w)$.
\end{enumerate}
Clearly, the algorithm terminates since the numbers in $D$ only become smaller.
Furthermore, the algorithm can be implemented in polynomial space since the number $M$ is exponentially bounded.

By~\cref{lem:mindip-rel} the algorithm maintains the property $D(A) \ge \mindip(A)$ for all $A \in N$.
We claim that, if the algorithm terminates, then in fact $D(A) = \mindip(A)$ holds for all $A \in N$.
Towards a contradiction, assume that $D(A) > \mindip(A)$, i.e.\ there exists a derivation $A \derivs u$
with $D(A) > \dip(u)$. Let us also assume that the derivation has minimal length.
Suppose that the first production which is applied in the derivation is $A \to \alpha_1 \cdots \alpha_n$,
and that $u = u_1 \cdots u_n$ where each $u_i$ is derived from $\alpha_i$.
We know that $D(\alpha_i) = \dip(u_i)$ by length-minimality of the derivation of $u$,
since $D(\alpha_i) > \dip(u_i)$ would imply that we could have chosen the derivation
$\alpha_i \derivs u_i$, which is smaller than $A \derivs u$, since the latter contains the former.
By termination we know that
\begin{align*}
  D(A) &\le D(\alpha_1 \cdots \alpha_n)\\
  &= \max_{1 \le k \le n} \left( D(\alpha_k) - \sum_{i=1}^{k-1} \offset(\alpha_i) \right) = \max_{1 \le k \le n} \left( \dip(u_k) - \sum_{i=1}^{k-1} \offset(u_i) \right) \stackrel{\eqref{eq:du}}{=} \dip(u),
\end{align*}
which is a contradiction.
\end{proof}

\subparagraph*{Dealing with undivided pumps}
Recall that a grammar $\cG$ is \emph{uniformly marked}
if $L(\cG)$ is contained in one of the subsets $\Theta^* \# \Theta^* \bar \# \Theta^*$, $\Theta^* \# \Theta^*$, $\Theta^* \bar \# \Theta^*$, or $\Theta^*$.
For such a grammar, we can partition its set of nonterminals $N$ into $N_{\# \bar \#} \cup N_\# \cup N_{\bar \#} \cup N_0$,
where $N_{\# \bar \#}$-nonterminals only produce marked words in $\Theta^* \# \Theta^* \bar \# \Theta^*$,
$N_\#$-nonterminals only produce marked words in $\Theta^* \# \Theta^*$, etc.
Furthermore, recall that a pump $A \derivs uAv$ is \emph{undivided} if $A \in N_{\# \bar \#} \cup N_0$, and \emph{divided} otherwise.

We would like to prove \cref{lem:convertTosECFG}, which would allow us to get rid of all
undivided pumps.
To this end, we still need to prove a few auxiliary results, the first of which is the following:

\dipsforspawns*

\begin{proof}

It is easy to construct a grammar $\cG_{A,a}$ with
$L(\cG_{A,a}) = \{ u \$ v \mid A \derivs u A v$ and $u$ contains $a \}$.
Since $\cG$ is tame-pumping, $\cG_{A,a}$ has uniform offset 0.
By \cref{lem:PAformulaEffect} we can compute in $\PSPACE$ a polynomially-sized existential Presburger formula
$\varphi$ for the relation
\[
	R = \bigcup_{u \$ v \in L(\cG_{A,a})} \psi(u) \times \psi(v).
\]
We only need to check whether there exists a tuple $(d_u, o_u, d_v, o_v) \in R$
with $d_u \le d_\markl$ and $d_v \le d_\markr$ (and possibly $o_u > 0$).
This concludes the proof since the truth problem of existential Presburger arithmetic is in $\NP$ \cite{BoroshTreybig76}.
\end{proof}

Let $L$ be a language.
We say that $L$ can be \emph{non-deterministically computed in $\PSPACE$}
if there is a non-deterministic polynomial space procedure that guesses a word,
and $L$ is the set of all such possible guesses.

\begin{lemma} \label{lem:divided-pump-dcl}
  Let $\cG = (N,\Sigma,P,S)$ be a uniformly marked tame-pumping $\CFG$.
  \begin{enumerate}
    \item Let be $A \in N$ be a nonterminal in $N_{\# \bar \#}$.
    Consider the language $L_A = \{u\#\bar\#v \mid A \derivs_\cG uAv\}$,
    which essentially contains the pumps of $A$.
    We can non-deterministically compute a language $L \subseteq \Theta^*\#\bar\#\Theta^*$
    in $\PSPACE$ such that $L$ has uniform offset $0$ and the following holds:
    \begin{enumerate}
      \item for every word $u\#\bar\#v \in L_A$ there is a word $u'\#\bar\#v' \in L$
      such that $\dip(u) \geq \dip(u')$ and $\dip(v) \geq \dip(v')$ as well as
      $\pi_\Gamma(u) \subword \pi_\Gamma(u')$ and $\pi_\Gamma(v) \subword \pi_\Gamma(v')$.
      \item for every word $u'\#\bar\#v' \in L$ there is a word $u\#\bar\#v \in L_A$
      such that $\dip(u') \geq \dip(u)$ and $\dip(v') \geq \dip(v)$ as well as
      $\pi_\Gamma(u') \subword \pi_\Gamma(u)$ and $\pi_\Gamma(v') \subword \pi_\Gamma(v)$.
    \end{enumerate}
    \item Let be $A \in N$ be a nonterminal in $N_0$.
    We can non-deterministically compute a language $L \subseteq \Theta^*$ in $\PSPACE$
    such that the following holds:
    \begin{enumerate}
      \item for every word $u \in L(\cG,A)$ there is a word $u' \in L$
      such that $u \extsw' u'$, and
      \item for every word $u' \in L$ there is a word $u \in L(\cG,A)$
      such that $u' \extsw' u$.
    \end{enumerate}
  \end{enumerate}
\end{lemma}

\begin{proof}
  \emph{Part 1.}
  For a pump of the form $u\#\bar\#v \in L_A$ we would like to compute a word $u'\#\bar\#v'$
  of the form $\bar{x}^{d_l}x^{d_l}\Gamma_l^*\#\bar\#\bar{x}^{d_r}x^{d_r}\Gamma_r^*$, where
  $\dip(u) = d_l$, $\dip(v) = d_r$, $\pi_\Gamma(u) \in \Gamma_l^*$, $\pi_\Gamma(v) \in \Gamma_r^*$,
  $\pi_{\Gamma\setminus\Gamma_l}(u) = \varepsilon = \pi_{\Gamma\setminus\Gamma_r}(v)$.
  In other words $u'\#\bar\#v'$ has offset $0$, involves the exact same handler names from $\Gamma$
  on both sides of the pump as $u\#\bar\#v$, and also has the exact same dips on both sides.
  It is then clear that for some choices of the infixes in $\Gamma_l^*$ and $\Gamma_r^*$ we have
  $\pi_\Gamma(u) \subword \pi_\Gamma(u')$ and $\pi_\Gamma(v) \subword \pi_\Gamma(v')$,
  yielding subcase (a).
  Moreover we clearly have $\pi_\Gamma(u') \subword \pi_\Gamma(u^n)$ and
  $\pi_\Gamma(v') \subword \pi_\Gamma(v^n)$ for a sufficiently large number $n$, e.g. $n = |u'v'|$.
  Since $u^n$ (respectively $v^n$) has the same dip as $u$ (respectively $v$) due to tame-pumping,
  this yields subcase (b).
  
  Let us now explain how to compute in $\PSPACE$ all $4$-tuples $d_l$, $d_r$, $\Gamma_l$, $\Gamma_r$
  that occur together as described above for some pump of $A$ in $G$.
  We begin by guessing such a $4$-tuple, which only takes up polynomial space:
  according to \cref{lem:checkTamePumping}~(2)
  the dips of words derived by a tame-pumping $\CFG$ are exponentially bounded,
  and therefore the numbers $d_l$ and $d_r$ take only polynomially many bits to write down.
  For each $a \in \Gamma_l$, we then use \cref{dips-for-spawns} to check in $\PSPACE$
  whether a pump $A \derivs_\cG u'au''Av$ exists such that
  $\dip(u'au'') \leq d_l$ and $\dip(v) \leq d_r$.
  We do the same for each $b \in \Gamma_r$ and a pump $A \derivs_\cG u'Av'bv''$ with 
  $d(u') = d(u)$ and $d(v'bv'') = d(v)$, which can also be done by symmetry and \cref{dips-for-spawns}.
  If all pumps exist, then the $4$-tuple can actually occur in a single pump
  (constructed by applying all the obtained pumps one after the other).
  On the other hand, if a single pump exists that matches the guessed $4$-tuple,
  then this pump serves as a witness for every $a \in \Gamma_l$ and $b \in \Gamma_r$,
  which is then recognized by the procedure.
  
  To briefly address the inequalities regarding the dip values in subcases (a) and (b):
  on one hand, for subcase (a), it is clearly possible to guess $d_l$ and $d_r$ such that they match
  $\dip(u)$ and $\dip(v)$ exactly, meaning a correct guess exists.
  On the other hand, for subcase (b), if we guess values $d_l$ and $d_r$ that are higher than the dips
  of any actually occurring pump of $A$, then the inequality still holds.
  
  \emph{Part 2.}
  Without loss of generality let $\cG$ be in Chomsky normal form.
  Consider a derivation tree of $\cG$ with root $A$ that does not contain any pumps.
  We traverse such a tree by performing a depth-first search and guessing a derivation in each step.
  Whenever we guess a derivation, every nonterminal occurring as an ancestor of the current node is
  not allowed to appear on the right hand side of said derivation, to ensure no pumps occur.
  Whenever we first explore a new node labelled by a nonterminal $B$, we also compute a $4$-tuple
  $d_l$, $d_r$, $\Gamma_l$, $\Gamma_r$ for one of $B$'s pumps, similar to Part 1, and we store the tuple
  as an additional label to this node.
  Then we output the left side of the pump, namely $\bar{x}^{d_l}x^{d_l}\Gamma_l^*$, and continue with
  its left child.
  After exploring all descendants and returning to the node, we output the right side of the pump,
  namely $\bar{x}^{d_r}x^{d_r}\Gamma_r^*$.
  If a node is labelled by a terminal, we simply output that terminal.
  
  We slightly alter this procedure whenever we encounter an increasing pump.
  Let us first observe that we can actually check this.
  When we compute the $4$-tuple like in Part 1, we already repeatedly use \cref{dips-for-spawns}
  to check for pumps, and said lemma allows us to furthermore check whether increasing pumps with
  the same properties exist.
  Moreover, when we combine several pumps, if any of them was increasing, then so is the combined pump.
  Now, if an increasing pump is found, we handle the current node slightly differently,
  outputting $\bar{x}^{d_l}x^{d_l}x^{D}\Gamma_l^*$ as the left side of its pump,
  and $\bar{x}^{D}\bar{x}^{d_r}x^{d_r}\Gamma_r^*$ as the right side, where $D = p(\cG)$
  is the bound on all dip values from \cref{lem:checkTamePumping}~(2) ($p$ is a polynomial).
  Then for all descendants of the node with the increasing pump,
  we stop computing dip values $d_l$ and $d_r$.
  This means below this node we only store $2$-tuples $\Gamma_l$, $\Gamma_r$, and output words
  $\Gamma_l^*$ on first exploration, and $\Gamma_r^*$ on final visit.
  Pumps that we check for here do not need to match any specific dip values.
  
  Let us now do a space analysis of this procedure.
  While performing the depth-first search, we store a path of nodes whose length is bounded by
  the height of the explored derivation tree $T$.
  Since we ensure that $T$ does not contain pumps, this is bounded by $|N|$.
  Each node is labelled by a single symbol in $N \cup \Sigma$ and possibly up to a $6$-tuple
  of subalphabets and numbers.
  The subalphabets have size bounded by $|\Gamma|$ and the numbers require polynomially many
  bits to write down, as discussed above and in Part 1.
  Finally the computation performed at each node is also possible in $\PSPACE$, as it matches Part 1.
  
  It remains to prove subcases (a) and (b).
  For subcase (a), consider a derivation tree of a word $u \in L(\cG,A)$.
  We replace every pump in this tree like in the above procedure to obtain a word $u' \in L$.
  Similar to Part 1 (a), we can ensure that we replace every infix with a word that has a potentially
  smaller dip and is larger in the subword ordering $\subword$.
  If we replace a pump below an increasing pump, the new dips are automatically smaller, as we set
  them to zero in the above procedure.
  Moreover the due to how $D$ is defined, when we replace an increasing pump, the word cannot dip
  below zero in between anymore.
  Due to tame-pumping, all non-increasing pumps have offset $0$ on both sides, like their replacements.
  This ensures that throughout the whole word, we can dip at most as far as in the original word.
  The offset requirements for $u \extsw' u'$ are also met, since we replace tame pumps with tame pumps,
  which contribute $0$ to the overall offset.
  
  For subcase (b), consider a word $u' \in L$.
  Like in Part 1 (b), we can simply switch the replacement pumps for original ones that have potentially
  smaller dips on both sides and are larger in the subword ordering $\subword$.
  For increasing pumps we can also repeat the original so often ($D$ times),
  that all dips below it do not matter.
  This results in a word $u \in L(\cG,A)$ with $u' \extsw u$.
\end{proof}

\converttosecfg*

\begin{proof}
  The idea is that $\cG'$ uses its nonterminals as a $\PSPACE$ tape, in accordance with \cref{ECFG-PSPACE}.
  This tape always contains a current nonterminal $A$ of $\cG$, for which $\cG'$ guesses
  the next derivation step $A \to BC$.
  If $A \in N_{\#\bar\#}$ then $\cG'$ first abstracts away a pump of $A$ via \cref{lem:divided-pump-dcl}.
  Then if $B \in N_0$ or $C \in N_0$, $\cG'$ abstracts away the entire derivation tree below
  said nonterminal by also using \cref{lem:divided-pump-dcl}.
  In both cases $\cG'$ non-deterministically guesses a single word in $\PSPACE$ and then derives it to the side of the current derivation tree, one letter at a time.
  This way, the second requirement of almost-pumpfree is met (see \cref{defn:UBFtree}).
  Afterwards $\cG'$ remembers on the tape that $A$ is not allowed to occur below, since
  all its pumps have already been abstracted.
  Then it continues with $B$ or $C$ or both, since at least one of them still produces a marker.
  For nonterminals in $N_\#$ or $N_{\bar\#}$ no pumps are abstracted away, since these
  would be divided pumps.
  
  Let us go into more detail.
  Let $\cG = (N,\Sigma,P,S)$ be wlog.\ in Chomsky normal form.
  On its $\PSPACE$ tape (see \cref{ECFG-PSPACE}), $\cG'$ always stores
  a set of disallowed nonterminals $M \subseteq N$ starting with $M = \emptyset$, and
  a current nonterminal $A \in N$ starting with $A = S$.
  When not in the middle of a computation, $\cG'$ begins makes a case distinction
  based on whether $A$ is in $N_{\#\bar\#}$ or in $N_\# \cup N_{\bar\#}$.
  Let us consider the latter case first.
  
  If $A \in N_\# \cup N_{\bar\#}$, $\cG'$ guesses the next production rule in $P$ to apply.
  In case of a rule $A \to \#$ or $A \to \bar \#$, $\cG'$ simply derives the corresponding marker.
  Otherwise it guesses a rule of the form $A \to BC$, where either $B \in N_0$ or $C \in N_0$.
  We consider $B \in N_0$ with the other case being symmetric.
  Here, $\cG'$ guesses a word in $\PSPACE$ using Part 2 of \cref{lem:divided-pump-dcl},
  with a slightly modified computation:
  instead of guessing words in $\Gamma_{\markl}$ and $\Gamma_{\markr}$,
  we simply output the whole alphabet both times.
  Then for every letter (and alphabet) computed this way, $\cG'$ derives a nonterminal
  to the left of the current node in the tree, which just produces this single letter
  (or alphabet, via an extended production).
  This is performed step by step, for one letter or alphabet at a time,
  during the $\PSPACE$ computation of \cref{lem:divided-pump-dcl}.
  Note that we do this to the left, because $B$ appears left in the production $A \to BC$.
  Afterwards we store $C$ as the current nonterminal and continue.
  
  If $A \in N_{\#\bar\#}$, $\cG'$ first sets $M$ to $M \cup \{A\}$
  and then guesses the next production rule $A \to BC$ to apply,
  ensuring that $B,C \notin M$.
  Now consider the case where $b$ is false.
  Then $\cG'$ uses Part 2 of \cref{lem:divided-pump-dcl} to compute a word of the form
  $\bar{x}^{d_l}x^{d_l}\Gamma_l^*\#\bar\#\bar{x}^{d_r}x^{d_r}\Gamma_r^*$ in $\PSPACE$.
  If this word corresponds to an increasing pump, $b$ is set to true.
  The computation of this word is similarly modified as above, computing whole alphabets.
  Hereby $\cG'$ begins by deriving letters and alphabets to the left, like above, until
  it would derive the markers $\#\bar\#$.
  These are not derived, and instead $\cG'$ switches to deriving the remaining objects to the right.
  Due to the shape of derivation trees, the computation also has to be modified so that
  the objects to the right are derived in reverse.
  Afterwards $\cG'$ continues for $B$ and $C$ like in the previous case,
  unless we have $B \in N_\#$ and $C \in N_{\bar\#}$.
  In the latter case, $\cG'$ simply continues on two different paths,
  one with $B$ and one with $C$ as the current nonterminal.
  
  The final case to consider is $A \in N_{\#\bar\#}$, where $\cG'$ has already guessed
  a production rule and $b$ is true.
  This case is very similar to the previous one except for one slight difference.
  When applying Part 2 of \cref{lem:divided-pump-dcl} to compute a word of the form
  $\bar{x}^{d_l}x^{d_l}\Gamma_l^*\#\bar\#\bar{x}^{d_r}x^{d_r}\Gamma_r^*$, we modify the
  computation even further:
  we drop the infixes corresponding to dips so that only the word
  $\Gamma_l^*\#\bar\#\Gamma_r^*$ is computed.
  
  Now we need to argue that the grammar $\cG'$ is as desired.
  Let $N'$ be the set of nonterminals of $\cG'$.
  It is clear from construction that $\cG'$ is almost-pumpfree:
  (1) the set $M$ ensures that a nonterminal in $N'_{\#\bar\#}$ cannot occur twice
  on the same path in a derivation tree of $\cG'$;
  (2) all nonterminals in $N'_0$ lead directly to leaves in the derivation tree.
  It remains to prove that $\dc{L(\cG)}=\dc{L(\cG')}$.
  
  For the $\subseteq$-direction, consider word $w \in L(\cG) \cap \AdmDecomps$
  and its derivation tree.
  Now invoke \cref{lem:divided-pump-dcl} to replace every undivided pump in $\outside{w}$ and every subtree below a $N_0$-nonterminal.
  This yields a word $w' \in L(\cG')$ by construction of $\cG'$.
  By subcase (a) in both parts of \cref{lem:divided-pump-dcl} we can choose the replacements in such
  a way that each new infix $v'$ of $w'$ replaces an old infix $v$ of $w$ with $v \extsw' v'$.
  Since the ordering $\extsw'$ is compatible with concatenation (see proof of \cref{lem:extsw-compatible}),
  we get $\outside{w} \extsw' \outside{w'}$ and $\inside{w} \extsw' \inside{w'}$.
  The words $w$ and $w'$ also contain the same markers, since we did not replace them.
  Therefore for $w \extsw w'$ we now only need to show that $w'$ is admissible, i.e.\ in $\AdmDecomps$.
  This follows from the fact that if $\inside{w}$ is an infix of, a prefix of, or an entire Dyck word,
  then so is $\inside{w'}$: the ordering $\sdyck$ preserves these properties, and is part of $\extsw'$.
  
  For the $\supseteq$-direction, consider word $w' \in L(\cG') \cap \AdmDecomps$
  and its derivation tree.
  Now consider every subtree, where $\cG'$ computes a replacement for either a pump of a nonterminal in $N_{\#\bar\#}$, or an entire subtree below a nonterminal in $N_0$.
  In the former case, replace it by an actual pump of said nonterminal in $\cG$,
  and in the latter case, replace it by an actual derivation tree of $\cG$ with said nonterminal as the root.
  By construction of $\cG'$, we obtain a tree in this way, whose yield is a word $w \in L(\cG)$.
  Now since each replacement computed by $\cG'$ is according to \cref{lem:divided-pump-dcl},
  when we switch these replacements for original derivations of $\cG$, we can always choose them
  according to subcase (b) in both parts of said lemma.
  This means each infix $v'$ of $w'$ that is switched in this way becomes an infix $v$ of $w$
  with $v' \extsw' v$.
  From here the proof is analogous to the $\subseteq$-direction.
  
  Finally, we need to argue that $|\cG'|$ is polynomially bounded in $|\cG|$.
  Every $\PSPACE$ procedure implemented by $\cG'$ uses a fixed number of states and tape symbols,
  which result in a fixed number of symbols added to the alphabet encoding the nonterminals of $\cG'$.
  The transducers then just need to check adjacent Turing machine configurations, which also requires
  a fixed amount of states.
  The length of the nonterminals of $\cG'$ is bounded by the largest polynomial bounding the tape length 
  of one of the implemented $\PSPACE$ procedures.
\end{proof}

\section{Results from Section~\ref{sec:absSplitPumps}}
\label{app:dividedPumps}

\subparagraph{Construction of $\calT_A$}
Let us first define transducers for the purposes of the proof of \cref{lem:sECFGtoDCAut}. Traditionally, a transducer is a finite-state machine where each edge can read a pair of words. For our construction, it will be convenient to extend the syntax slightly: We allow transitions where some component is $\Xi^*$ for some alphabet $\Xi$. The semantics is the obvious one: Such a transition allows the transducer to read any word over $\Xi$.

Formally, a transducer is a tuple $\calT=(Q,\Sigma,E,q_0,F)$, where $Q$ is its
finite set of \emph{states}, $\Sigma$ is its input alphabet, $E$ is its finite
set of \emph{edges}, $q_0\in Q$ is its \emph{initial state}, and $F\subseteq Q$
is the set of \emph{final states}. An edge is of the form
$p\xrightarrow{(u,v)}q$ for
$(u,v)\in\Sigma_\varepsilon\times\Sigma_\varepsilon$, where
$\Sigma_\varepsilon=\Sigma\cup\{\varepsilon\}$, or of the form
$p\xrightarrow{(\Xi_1,\Xi_2)}q$ for some $\Xi_1,\Xi_2\subseteq\Sigma$. 

To describe the semantics, we define an induced step relation $\xRightarrow{\cdot}$.
For words $u,v\in\Sigma^*$, we write $p\xRightarrow{(u,v)}q$ if and only if
\begin{enumerate}
\item there exists an edge $p\xrightarrow{(u,v)}q$ or 
\item there exists an
edge $p\xrightarrow{(\Xi_1,\Xi_2)}q$ such that $u\in\Xi_1^*$ and $v\in\Xi_2^*$.
\end{enumerate}
A pair $(u,v)\in\Sigma^*\times\Sigma^*$ is \emph{accepted} by $\calT$ if there
are states $q_1,\ldots,q_n\in Q$ with $q_n\in F$ and pairs
$(u_1,v_1),\ldots,(u_n,v_n)\in\Sigma^*\times\Sigma^*$ such that
$q_i\xRightarrow{(u_{i+1},v_{i+1})} q_{i+1}$ for each $i\in[0,n-1]$ and
$u=u_1\cdots u_n$ and $v=v_1\cdots v_n$.

In a \emph{succinct} transducer, the set of states is the set $\Lambda^M$ for
some alphabet $\Lambda$ and some number $M$ specified in unary. Moreover, the
edges are specified using (i)~a (non-succinct) transducer for each pair
$(u,v)\in\Sigma_\varepsilon\times\Sigma_\varepsilon$ which reads the set of all
pairs $(p,q)$ such that there exists an edge $p\xrightarrow{(u,v)}q$ and (ii)~a
(non-succinct) transducer that describes the edges
$p\xrightarrow{(\Xi_1,\Xi_2)}q$. More precisely, for
$\Xi_1,\Xi_2\subseteq\Sigma$, the latter transducer accepts the pair
$(px_1,qx_2)$ with $p,q\in\Lambda^M$ and $x_1,x_2\in\Sigma^{\le|\Sigma}$ if and
only if there exists an edge $p\xrightarrow{(\Xi_1,\Xi_2)}q$, where $\Xi_i$ is
the set of letters appearing in $x_i$, for $i\in\{1,2\}$.

We now describe the transducer $\calT_A$. It will be clear from the
construction that an equivalent polynomial size succinct transducer can be
constructed. The set of states of $\calT$ is the set of non-terminals of $\calG$.
It has the following edges. For each production $A\to BC$ in $\calG$, we have the edges
\begin{align*}
&A\xrightarrow{(\Xi,\emptyset)}C & &\text{for each production $B\to \Xi^*$ in $\calG$, $\Xi\subseteq\Sigma$} \\
&A\xrightarrow{(a,\varepsilon)}C & &\text{for each production $B\to a$ in $\calG$, $a\in\Theta$}, \\
&A\xrightarrow{(\emptyset,\Xi)}B & &\text{for each production $C\to \Xi^*$ in $\calG$, $\Xi\subseteq\Sigma$}, \\
&A\xrightarrow{(\varepsilon,a)}B & &\text{for each production $C\to a$ in $\calG$, $a\in\Theta$}.
\end{align*}
Moreover, $A$ is the initial state and the only final state. Then $\calT$ clearly has the desired properties.

\subparagraph{Bounding offset and dip of divided pumps}
It is a direct consequence of \cref{lem:boundingOffsetInUPFtree} that in a
derivation of an admissible word, the offset and dip of words $u$ and $v$ that
occur in a pump $A\derivs uAv$ must be bounded doubly exponentially:
\begin{lemma}\label{lem:boundingPumpPrefixes}
There exists a polynomial $q$ such that for any uniformly marked, tame-pumping,
almost-pumpfree $\sECFG$ $\calG$, the following holds. Let $A\derivs uAv$ be a
pump that appears in the derivation of an admissible word, such that $A\in
N_{\#}\cup N_{\bar{\#}}$. Then for every suffix $u'$ of $u$ and every prefix
$v'$ of $v$, we have $|\offset(u'^{\rev})|,|\dip(u'^{\rev})|,|\offset(v')|,|\dip(v')|\le 2^{q(|\calG|)}$.
\end{lemma}
\begin{proof}
	Let $p$ be the polynomial from \cref{lem:boundingOffsetInUPFtree}. 
	Observe that it suffices to show the bound for the offsets: Since we prove it for all suffixes/prefixes, the bound on dips is implied.

i%
	Consider a derivation tree $T$ with root label $B$ such that (i)~the pump $A\derivs uAv$ occurrs in $T$ and (ii)~$T$ derives $B\derivs xuwvy$, where $xuwvy$ is admissible.
	Without loss of generality, suppose $w$ belongs to $\Theta^*\#\Theta^*$ (the case of $\Theta^*\bar{\#}\Theta^*$ is analogous) and write $w=r\#s$.
	Since $\calG$ is almost-pumpfree, we have subtrees $T'$ and $T_0$ of $T$ such that $T'$ derives $u'wv'=u'r\#sv'$ and $T_0$ derives $w=r\#s$.
	By the choice of $p$, we now have $|\offset(u'r)|,|\offset(r)|,|\offset(v')|,|\offset(sv')|\le 2^{p(|\calG|)}$. This implies
	\begin{align*}
		|\offset(u')|&=|\offset(u'r)-\offset(r)|\le |\offset(u'r)|+|\offset(r)|\le 2^{p(|\calG|)+1}, \\
		|\offset(v')|&=|\offset(sv')-\offset(s)|\le |\offset(sv')|+|\offset(s)|\le 2^{p(|\calG|)+1}.
	\end{align*}
	Hence, setting $q(n)=p(n)+1$ yields the result.
\end{proof}

\subparagraph{Construction of $\calT_{A,\bx}$}
According to \cref{lem:boundingPumpPrefixes}, for some given
$\bx=(d_\markl,\delta_\markl,d_\markr,\delta_\markr)$ with
$d_\markl,d_\markr\in[0,2^{q(|\calG|)}]$ and
$\delta_\markl,\delta_\markr\in[-2^{q(|\calG|)},2^{q(|\calG|)}]$, we can now
turn each transducers $\calT_A$ into a transducer $\calT_{A,\bx}$ which accepts
a pair $(u,v)$ only if $(u,v)$ is accepted by $\calT_A$ and it also satisfies
$e(u^{\rev})=(d_{\markl},\delta_{\markl})$ and
$e(v)=(d_{\markr},\delta_{\markr})$. The transducer $\calT_{A,\bx}$ has states
$(q,\by)$, where $q$ is a state of $\calT_A$ and
$\by=(d'_\markl,\delta'_\markl,d'_\markr,\delta'_\markr)$ with
$d'_\markl,d'_\markr\in[0,2^{q(|\calG|)}]$ and
$\delta'_\markl,\delta'_\markr\in[-2^{q(|\calG|)}]$. When reading a pair
$(u,v)$, it simulates $\calT_A$ in the component $q$ and it stores $e(u)$ and
$e(v)$ in the component $\by$. Here, the bounds for $\by$ are sufficient because
\cref{lem:boundingPumpPrefixes} tells us that for any such pair, the offset and
dip will remain in the respective interval. It is clear that we can construct a succinct polynomial-size presentation for each $\calT_{A,\bx}$.

\subparagraph{Ideals and skeleton runs of transducers}
Let us define ideals over an alphabet $\Sigma$. An \emph{atom} is a set of the
form $\{a,\varepsilon\}$ for $a\in\Sigma$ or a set of the form $\Xi^*$ for some
$\Xi\subseteq\Sigma$.  An \emph{ideal} is a finite product $A_1\cdots A_n$ of
atoms. Observe that each ideal can be written as a product
$\Xi_0^*\{a_1,\varepsilon\}\Xi_1^*\cdots \{a_n,\varepsilon\}\Xi_n^*$ for
$a_1,\ldots,a_n\in\Sigma$ and $\Xi_0,\ldots,\Xi_n\subseteq\Sigma$.

Let $\calT_{A,\bx}$ be a transducer as constructed above. Recall that its input alphabet is $\Theta=\Gamma\cup\{x,\bar{x}\}$. A \emph{skeleton run} is a sequence of edges $q_0\xrightarrow{(x_1,y_1)}q_1, q_1\xrightarrow{(x_2,y_2)}q_2, \ldots, q_{n-1}\xrightarrow{(x_n,y_n)}q_n$ such that the states $q_0,\ldots,q_n$ are pairwise distinct. To this skeleton run, we associate its \emph{left ideal} as $\Gamma_0^*A_1\Gamma_1^*\cdots A_n\Gamma_n^*$, where
\begin{enumerate}
\item for each $i\in[0,n]$, $\Gamma_i\subseteq\Gamma$ is the set of letters from $\Gamma$ that occur in some left component on a cycle of $\calT$ from $q_i$ to $q_i$, and
\item for each $i\in[1,n]$, $A_i$ is the following atom:
	\begin{enumerate}
		\item If the edge $q_{i-1}\xrightarrow{(x_i,y_i)}q_i$ is of the form $(x_i,y_i)\in\Sigma_\varepsilon\times\Sigma_\varepsilon$, then $A_i=\{x_i,\varepsilon\}$.
		\item If the edge $q_{i-1}\xrightarrow{(x_i,y_i)}q_i$ is of the form $q_{i-1}\xrightarrow{(\Xi,\Xi')}q_i$, then $A_i=\Xi^*$.
	\end{enumerate}
\end{enumerate}
The \emph{right ideal} of the skeleton run is defined by taking the right
components instead of the left components when specifying $\Gamma_i$ and $A_i$.

Observe that in a succinct transducer, a skeleton run has at most exponential
length. Moreover, given such a skeleton run, one can compute its (exponentially
long) left ideal and its right ideal using polynomial space (in the size of the
transducer): The alphabets $\Xi_i$ can be computed by simulating cycles of the
transducer, which only requires polynomial space.

\label{afterbibliography}
\newoutputstream{pagestotal}
\openoutputfile{main.pagestotal.ctr}{pagestotal}
\addtostream{pagestotal}{\getpagerefnumber{afterbibliography}}
\closeoutputstream{pagestotal}

\end{document}